\documentclass{amsart}
\theoremstyle{definition}
\newtheorem{definition}{Definition}[section]

\theoremstyle{plain}
\newtheorem{theorem}{Theorem}[section]

\newtheorem{corollary}[theorem]{Corollary}
\newtheorem{lemma}[theorem]{Lemma}
\newtheorem{proposition}[theorem]{Proposition}

\theoremstyle{remark}
\newtheorem{remark}{Remark}[section]

\newcommand{\bbE}{\mathbb{E}}
\newcommand{\bbP}{\mathbb{P}}
\newcommand{\bbR}{\mathbb{R}}
\newcommand{\bbZ}{\mathbb{Z}}

\newcommand{\calE}{\mathcal{E}}

\newcommand{\calN}{\mathcal{N}}
\newcommand{\calW}{\mathcal{W}}
\newcommand{\calZ}{\mathcal{Z}}
\newcommand{\fp}{\mathfrak{p}}

\newcommand{\zv}{\boldsymbol{0}}
\newcommand{\HS}{\mathcal{HS}}
\newcommand{\Z}{\mathbb{Z}}

\newcommand{\corr}{\operatorname{corr}}
\newcommand{\Cov}{\operatorname{Cov}}
\newcommand{\Pois}{\operatorname{Pois}}

\newcommand{\Var}{\operatorname{Var}}
\newcommand{\wt}{\operatorname{wt}}

\usepackage{hyperref}
\usepackage{amsthm,mathtools}
\usepackage[all]{xy}
\usepackage{mathrsfs}
\usepackage[dvipsnames]{xcolor}
\usepackage{amssymb,amsmath,tabularx,graphicx,float,placeins}
\usepackage{tabu} 
\usepackage{makecell} 
\usepackage{esint} 

\usepackage{tikz}
\usetikzlibrary[calc,intersections,through,backgrounds,arrows,decorations.pathmorphing]

\newcommand{\defn}[1]{{\bf \textcolor{blue}{#1}}}

\begin{document}

\title{Gaussian Broadcast on Grids}
\author{Pakawut Jiradilok and Elchanan Mossel}
\date{\today}

\address{Department of Mathematics, Massachusetts Institute of Technology, Cambridge, MA 02139}
\email[P.~Jiradilok]{pakawut@mit.edu}

\address{Department of Mathematics, Massachusetts Institute of Technology, Cambridge, MA 02139}
\email[E.~Mossel]{elmos@mit.edu}

\begin{abstract}
Motivated by the classical work 
on finite noisy automata (Gray 1982, G\'{a}cs 2001, Gray 2001) and by the recent work on broadcasting on grids (Makur, Mossel, and Polyanskiy 2022), we 
introduce Gaussian variants of these models. 
These models are defined on graded (layered) posets. 
At time $0$, all the nodes at layer $0$ begin with the same random variable $X_0$. At time $k\geq 1$, each node in the $k^{\text{th}}$ layer computes a linear combination 
of its inputs at layer $k-1$ with independent Gaussian noise added. 
The main question is: when is it possible to recover the original signal $X_0$ with non-vanishing correlation? We consider different notions of recovery including recovery from a single node, recovery from a bounded window (with convex or general) coefficients, and recovery from an unbounded window (with convex or general) coefficients. 

Our main interest is in two models defined on grids: 
\begin{itemize}
\item 
The infinite model. Here 
layer $k$ is given by the vertices of $\Z^{d+1}$ whose sum of coordinates is $k$ and for a vertex $v$ at layer $k \geq 1$, we let 
$X_v = \alpha \sum_{u \in \fp(v)} (X_u + W_{u, v})$, where $\fp(v)$ is the set of vertices $u$ in 
layer $k-1$ that differ from $v$ exactly in one coordinate, and $W_{u,v}$ are independent standard normal. 

We show that when $\alpha < \frac{1}{d+1}$, the correlation between $X_v$ and $X_0$ decays exponentially to $0$ as $|v| \to \infty$, and when $\alpha > \frac{1}{d+1}$, the correlation is bounded away from $0$. The critical case when $\alpha = \frac{1}{d+1}$ exhibits a phase transition in dimension, where $X_v$ has non-vanishing correlation with $X_0$ if and only if $d \geq 3$. The same results hold for any bounded window. 

\item 
The finite model. Here layer $k$ is given by the vertices of $\Z^{d+1}$ with nonnegative coordinates whose sum is $k$. In this case, we identify the sub-critical and the super-critical regimes. In the sub-critical regime, the correlation decays to zero for unbounded windows. In the super-critical regime, there exists for every $t$ a convex combination of $X_u$ (across the vertices $u$ on layer $t$) whose correlation is bounded away from $0$. Interestingly, we find that for the critical parameters, the correlation is vanishing in all dimensions and for unbounded window sizes.  
\end{itemize}

Our results may elucidate the existing results on phase transitions in noisy finite automata and on broadcast models on grids. 

In particular, the classical finite automata models and our analogous critical infinite models exhibit the same phase transition in dimension. Moreover, our critical finite models may shed light on the broadcast model on grids, as our results establish that the correlation vanishes at all dimensions.
\end{abstract}

\maketitle

\newpage

\section{Introduction}\label{s:intro}
In this paper we introduce and study Gaussian versions of the classical noisy finite automaton model (see e.g.~\cite{Gra01}) and of the more recent broadcasting on grids model~\cite{MaMoPo:22}.  

The main question on cellular automaton is: when is it possible for noisy local computation to remember information indefinitely. 
There are many examples from statistical physics in dimension $2$ or more when local noisy computation remembers information indefinitely. Perhaps the simplest example is Glauber dynamics in two dimensions, where at low temperatures there are two different stationary measures, i.e., the plus and minus measures (see e.g.~\cite{Martinelli:99}). 

Gray made a very strong conjecture (see e.g.~\cite{Gra01}) that such memorization cannot hold in one dimension. Groundbreaking work by Gacs~\cite{Gac01} showed that the strong conjecture by Gray was in fact incorrect and that memorization is possible in one dimension. However, crucially the construction in~\cite{Gac01} required the update of each node to depend on a large but bounded neighborhood and more importantly in a non-monotone fashion. Earlier results of Gray~\cite{Gra82} showed that if the interaction is both monotone and nearest neighbor then memorization is not possible in one dimension (in continuous time).  

A new variant of the problem was introduced in~\cite{MaMoPo:22} where it is called {\em broadcasting on grids}. In the new variant, information is propagated
from $(0,0)$ on the positive quadrant in $\Z^2$ where each node $v$ computes a noisy function of its parent set, i.e., those vectors that have exactly one coordinate smaller than that of $v$. The authors of~\cite{MaMoPo:22} showed that in the new variant, there is no memorization. 
They conjectured that memorization is possible in $d \geq 3$ dimensions.

The models above all consider broadcast of discrete signals with discrete noise channels (i.e., the binary symmetric channel). It is often easier to study Gaussian models and our goal in this paper is to introduce and analyze such models. We are able to analyze the Gaussian models precisely. We hope this will allow to shed light on difficult open problems remaining in~\cite{Gac01,Gra01,MaMoPo:22}.

In our Gaussian models, nodes compute averages of the values of their parents to which independent Gaussian noise is added. This is perhaps the simplest and most straightforward analogue of the models above. 
We ask: when does memorization hold for such models?
\begin{itemize}
\item 
{\em The infinite model}. This model is the analogue of cellular automata on $\Z^d$~\cite{Gra82,Gra01}. In this model, layer $k$ is given by the vertices of $\Z^{d+1}$ whose sum of coordinates is $k$ and for a vertex $v$ at layer $k \geq 1$, we let 
$X_v = \frac{1}{d+1} \sum_{u \in \fp(v)} (X_u + W_{u,v})$, where $\fp(v)$ is the set of vertices $u$ in 
layer $k-1$ that differ from $v$ exactly in one coordinate, and $W_{u,v}$ are independent standard normal. 
Here we prove analogous results to what is known for monotone automata on $\Z^d$, i.e., memorization holds 
 if and only if $d \geq 3$.  

\item 
{\em The finite model}. This is the analogue of the broadcast model of~\cite{MaMoPo:22}. 
Here layer $k$ is given by the vertices of $\Z^{d+1}$ with nonnegative coordinates whose sum is $k$. Again, each node averages its parents with an addition of independent Gaussian noise. 
Here we show that for all dimensions there is no memorization. For $d=2$ this is analogous to the results of~\cite{MaMoPo:22}. However it is not analogous to the conjecture of~\cite{MaMoPo:22} that in their broadcast model memorization is possible if $d \geq 3$. 
\end{itemize}
The Gaussian model actually allows for update rules of the form 
\[
X_v = \alpha_{|\fp(v)|} \sum_{u \in \fp(v)} \left(X_u + W_{u,v}\right),
\]
where $\fp(v)$ is the set of vertices $u$ in 
layer $k-1$ that differ from $v$ exactly in one coordinate and we analyze the behavior of all such rules. The most interesting case is the critical case where $\alpha_{|\fp(v)|}$ is $1/|\fp(v)|$ for every $v$. 
However, our results cover all possible functions of $\alpha$.

\subsection{Other related models and results}
We may consider similar processes defined on rooted trees. This is done in the Gaussian case in~\cite{MoRoSl:13} which establishes the Kesten--Stigum (KS) bound as the critical threshold for the process. 
This is in contrast to discrete channels where the KS bound is only the threshold in some specific cases (see
~\cite{BlRuZa:95,Ioffe:96a,Ioffe:96b,Mossel:98,Mossel:01,PemantlePeres:10,Sly:09,Sly:09a,JansonMossel:04,MoSlSo:23,Mossel:23}).

The problems considered here are closely related to the problem of noisy computation~\cite{vonNeumann:56,EvansSchulman:99}. Indeed it can be thought of in the following way: suppose we want to remember a bit in a noisy circuit of depth $k$. How big should the circuit be? Von Neumann~\cite{vonNeumann:56} asked this question assuming we take multiple clones of the original bit and recursively apply gates in order to reduce the noise. The broadcasting model 
on trees perspective is to start from a single bit and repeatedly clone it so that one can recover it well from the nodes at depth $k$. The model we consider here can be viewed as a Gaussian version of this model. 

\bigskip

\section{Formal Definitions and Main Results}\label{s:defns-main-results}
\subsection{Preliminaries}
\subsubsection{Poset Terminologies}\label{subsubsec:poset-terms}
We begin by recalling a few useful poset terminologies. For a standard reference on this subject, we refer to \cite[Ch.~3]{Sta12}.

Let $(P,\le)$ be a poset. An element $u \in P$ is said to be \defn{minimal} if no element $v \in P$ satisfies $u > v$. For two poset elements $u, v \in P$, we say that $u$ \defn{is covered by} $v$ (or $v$ \defn{covers} $u$) if $u < v$ and no element $w \in P$ satisfies $u < w < v$. Following \cite{Sta12}, we denote this covering relation by $u \lessdot v$. A \defn{chain} $C$ in $P$ is an induced subposet of $P$ such that for any $u, v \in C$, we have either $u \le v$ or $u \ge v$ (i.e., any two elements in $C$ are comparable).

A \defn{maximal chain} is a chain which is not properly contained in another chain. A \defn{saturated chain} is a chain $C$ in the poset $P$ for which there do not exist $u,v \in C$ and $w \in P \setminus C$ such that $u < w < v$ and such that $C \cup \{w\}$ is a chain. For example, in the poset $P = \{00,01,10,11\}$ with four covering relations (i) $00 \lessdot 01$, (ii) $00 \lessdot 10$, (iii) $01 \lessdot 11$, and (iv) $10 \lessdot 11$, the chain $C_1 = \{10 \lessdot 11\}$ is saturated but not maximal, while the chain $C_2 = \{00 \lessdot 10 \lessdot 11\}$ is maximal (and hence saturated).

The poset $P$ is said to be \defn{graded} if $P$ can be decomposed into a disjoint union
\[
P = P_0 \uplus P_1 \uplus P_2 \uplus \cdots
\]
such that every maximal chain in $P$ is of the form
\[
u_0 \lessdot u_1 \lessdot u_2 \lessdot \cdots,
\]
where for each $i \ge 0$, $u_i \in P_i$. When $P$ is graded, we define the \defn{rank function} $\rho: P \to \mathbb{Z}_{\ge 0}$ so that $\rho(u) = i$ if $u \in P_i$, and in this case we say that the \defn{rank} of $u$ is $i$.

\subsubsection{Model and Definition of Reconstruction}

Let $d \in \mathbb{Z}_{\ge 0}$ be a nonnegative integer. Let $P$ be an infinite graded poset in which every element covers at most $d+1$ elements and is covered by at least one element.  Let $\alpha_1, \ldots, \alpha_{d+1}, \varepsilon, \mu_0, \sigma_0^2$ be real numbers such that $\alpha_1, \ldots, \alpha_{d+1}, \varepsilon, \sigma_0^2 > 0$. We let $X_0 \sim \calN(\mu_0, \sigma_0^2)$, and introduce a fresh (independent from $X_0$) set of i.i.d. standard Gaussian $W_{u,v} \sim \calN(0,1)$, indexed by pairs of elements $u, v \in P$ such that $u \lessdot v$. Throughout this paper, we may also use the notation $W_{u \to v}$ for the random variable $W_{u,v}$.

We associate to each element $v \in P$ a random variable $X_v$, defined inductively on the rank of $v$ as follows. For the elements $v \in P$ of rank $0$ (namely, the minimal elements of $P$), we let $X_v = X_0$.

For each element $v \in P$, let $\fp(v)$ denote the set of elements of $P$ covered by $v$. Note that since $P$ is graded, the set $\fp(v)$ is empty if and only if the rank of $v$ is $0$. For every element $v \in P$ with rank at least $1$, define $X_v$ by the \defn{model recurrence}
\begin{equation}\label{eq:model-rec}
X_v := \alpha_{|\fp(v)|} \sum_{u \in \fp(v)} \left( X_u + \varepsilon \cdot W_{u \to v} \right).
\end{equation}

By gradedness, the poset $P$ can be decomposed into \defn{layers} according to rank as
\[
P = L_0 \uplus L_1 \uplus L_2 \uplus \cdots,
\]
where the $t^{\text{th}}$ layer $L_t$ contains all the elements of $P$ of rank $t$.

\begin{definition}
A real-valued function $c: P \to \mathbb{R}$ is said to \defn{have finite support on each layer} if for each $t \ge 0$, the restriction
\[
c|_{L_t}: L_t \to \mathbb{R}
\]
has finite support. In other words, for each $t \ge 0$, we have that $c(u) = 0$ for all but finitely many $u \in L_t$.
\end{definition}

\begin{definition}\label{def:rec}
We say that \defn{reconstruction is possible} if there exists a function
\begin{align*}
& c: P \to \mathbb{R}, \\
& c: u \mapsto c_u,
\end{align*}
which has finite support on each layer such that the sequence $\{\zeta_t\}_{t=0}^{\infty}$ of random variables given by
\[
\zeta_t := \sum_{u \in L_t} c_u X_u,
\]
satisfies
\begin{itemize}
    \item $\forall t \in \mathbb{Z}_{\ge 0}, \Var(\zeta_t) > 0$, and
    \item the correlation
    \[
    \corr(\zeta_t, X_0) := \frac{\Cov(\zeta_t, X_0)}{\sqrt{\Var(\zeta_t) \Var(X_0)}}
    \]
    does not converge to $0$ as $t \to \infty$.
\end{itemize}

Furthermore, if the function $c: P \to \mathbb{R}$ can be chosen to have only nonnegative values, then we say \defn{convex reconstruction is possible}.
\end{definition}

Whether reconstruction is possible or not is a function of the poset and the parameters. Thus we also say that reconstruction is possible \defn{for the poset} $P$ \defn{and the parameter tuple} $(\alpha_1, \ldots, \alpha_{d+1}, \varepsilon, \mu_0, \sigma_0^2)$.

In this paper, we are particularly interested in posets whose elements are on the grid; namely, as a set, $P \subseteq \mathbb{Z}^{d+1}$ for some $d \ge 0$. For these posets, we give the following definitions.

\begin{definition}
A \defn{window} $\calW$ of width $N$ in $P$ is a subset of $P$ of the form
\[
\calW := \big( [x_1, x_1 + N) \times [x_2, x_2 + N) \times \cdots \times [x_{d+1}, x_{d+1} + N) \big) \cap P,
\]
for some integers $x_1, \ldots, x_{d+1}$.
\end{definition}

\begin{definition}
We say that \defn{local reconstruction is possible} if there exist a finite number $N$, a sequence $\{\calW_t\}_{t=0}^{\infty}$ of windows of width $N$ in $P$, and a function
\begin{align*}
& c: P \to \mathbb{R}, \\
& c: u \mapsto c_u,
\end{align*}
such that for each $t \ge 0$, the restriction
\[
c|_{L_t}: L_t \to \mathbb{R}
\]
is supported on $\calW_t$, and such that the sequence $\{\zeta_t\}_{t=0}^{\infty}$ of random variables given by
\[
\zeta_t := \sum_{u \in L_t} c_u X_u,
\]
satisfies
\begin{itemize}
    \item $\forall t \in \mathbb{Z}_{\ge 0}, \Var(\zeta_t) > 0$, and
    \item the correlation
    \[
    \corr(\zeta_t, X_0) := \frac{\Cov(\zeta_t, X_0)}{\sqrt{\Var(\zeta_t) \Var(X_0)}}
    \]
    does not converge to $0$ as $t \to \infty$.
\end{itemize}

In addition,
\begin{itemize}
    \item If the function $c: P \to \mathbb{R}$ can be chosen to have only nonnegative values, then we say \defn{local convex reconstruction is possible}.
    \item If $N = 1$, then we say \defn{single-vertex reconstruction is possible}.
\end{itemize}
\end{definition}

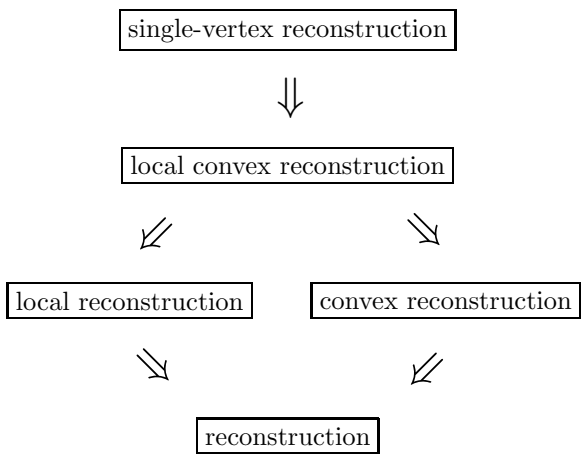
\begin{figure}
\begin{center}
\begin{tikzpicture}
\usetikzlibrary{arrows}
\begin{scope}[scale = 1.0, shift={(0,5.4)}]
\node at (0,0) {\boxed{\text{single-vertex reconstruction}}};
\end{scope}

\begin{scope}[scale = 1.0, shift={(0,4.5)}]
\node at (0,0) {\rotatebox{270}{\huge $\Rightarrow$}};
\end{scope}

\begin{scope}[scale = 1.0, shift={(0,3.6)}]
\node at (0,0) {\boxed{\text{local convex reconstruction}}};
\end{scope}

\begin{scope}[scale = 1.0, shift={(-1.8,2.7)}]
\node at (0,0) {\rotatebox{225}{\huge $\Rightarrow$}};
\end{scope}

\begin{scope}[scale = 1.0, shift={(1.8,2.7)}]
\node at (0,0) {\rotatebox{315}{\huge $\Rightarrow$}};
\end{scope}

\begin{scope}[scale = 1.0, shift={(-2.1,1.8)}]
\node at (0,0) {\boxed{\text{local reconstruction}}};
\end{scope}

\begin{scope}[scale = 1.0, shift={(2.1,1.8)}]
\node at (0,0) {\boxed{\text{convex reconstruction}}};
\end{scope}

\begin{scope}[scale = 1.0, shift={(-1.8,0.9)}]
\node at (0,0) {\rotatebox{315}{\huge $\Rightarrow$}};
\end{scope}

\begin{scope}[scale = 1.0, shift={(1.8,0.9)}]
\node at (0,0) {\rotatebox{225}{\huge $\Rightarrow$}};
\end{scope}

\begin{scope}[scale = 1.0, shift={(0,0)}]
\node at (0,0) {\boxed{\text{reconstruction}}};
\end{scope}
\end{tikzpicture}
\end{center}
\caption{Poset of relations.}\label{fig:poset-of-relations}
\end{figure}

In Figure~\ref{fig:poset-of-relations}, we show the poset of relations from the five definitions of reconstruction. In the diagram, $\boxed{A} \Rightarrow \boxed{B}$ means that if $A$ is possible, then $B$ is possible.

\subsection{Main Results}
In this subsection, we describe our main results. In Section~\ref{s:para-reduct} below, we show that whether each type of reconstruction is possible does not depend on the parameters $\varepsilon, \mu_0, \sigma_0^2$. Thus, let us assume for the rest of this section that $\varepsilon = 1$, $\mu_0 = 0$, and $\sigma_0^2 = 1$.

\subsubsection{The Finite Model}

We begin with the \defn{finite model} in which the poset is $P = \mathbb{Z}_{\ge 0}^{d+1}$ with the partial order given by $u \ge v$ if and only if $u - v \in \mathbb{Z}_{\ge 0}^{d+1}$, where $u-v$ denotes the entrywise subtraction. The elements in the $t^{\text{th}}$ layer $L_t$ are exactly those whose sums of entries are exactly $t$. In particular, there exists a unique minimal element $\zv \in L_0$.

\begin{theorem}\label{thm:rec-orthant}
Let $d \ge 0$. When $P = \mathbb{Z}_{\ge 0}^{d+1}$, we have the following.
\begin{itemize}
    \item[(a)] Convex reconstruction is possible if there exists $i \in [d+1]$ such that $\alpha_i > 1/i$.
    \item[(b)] Reconstruction is not possible if for all $i \in [d+1]$, $\alpha_i \le 1/i$.
\end{itemize}
\end{theorem}
We prove Theorem~\ref{thm:rec-orthant} in Section~\ref{s:orthant}.

To motivate our next main result, let us consider the quantity
\[
S_t := \sup \frac{\Cov(\zeta, X_0)}{\sqrt{\Var(\zeta)}},
\]
where the supremum is over all choices of
\[
\zeta = \sum_{u \in L_t} c_u X_u
\]
such that (i) the coefficients $c_u$ are not simultaneously zero, and (ii) all but finitely many coefficients $c_u$ are zero.

By definition, reconstruction is not possible if and only if converges to $0$ as $t \to \infty$. When reconstruction is not possible, it is natural to ask about the {\em rate} at which $S_t$ goes to $0$.

Let us refer to the case where $\alpha_i = 1/i$ for every $i \in [d+1]$ as the \defn{critical case}. We find that computations in the critical case when $d = 1$ are particularly nice. Our following result shows the exact formula for $S_t$ in this case.

\begin{theorem}\label{thm:exact-formula}
When $P = \mathbb{Z}_{\ge 0}^2$ and $(\alpha_1, \alpha_2) = (1,1/2)$, the random variable
\[
\widehat{\zeta} := \binom{t}{0} X_{(0,t)} + \binom{t}{1} X_{(1,t-1)} + \cdots + \binom{t}{t} X_{(t,0)}
\]
is a maximizer of the function
\[
\zeta \mapsto \frac{\Cov(\zeta, X_0)}{\sqrt{\Var(\zeta)}}
\]
among all nonzero linear combinations of $X_u$ with $u \in L_t$. Moreover, we have the exact formula
\[
\frac{\Cov(\widehat{\zeta}, X_0)}{\sqrt{\Var(\widehat{\zeta})}} = \frac{1}{\sqrt{\frac{t}{4^t}\binom{2t}{t} + 1}}.
\]
\end{theorem}
This result also shows the {\em exact} rate of convergence: $S_t \sim \sqrt[4]{\pi/t}$, as $t \to \infty$. We prove Theorem~\ref{thm:exact-formula} in Subsection~\ref{subsec:exact-formula} below.

For convex reconstruction in the critical case, we have the following result.
\begin{theorem}\label{thm:conv-corr}
Let $d \ge 0$. Let $P = \mathbb{Z}_{\ge 0}^{d+1}$ and $\alpha_i = 1/i$ for each $i \in [d+1]$. For any
\[
\zeta = \sum_{u \in L_t} c_u X_u,
\]
with $c_u \ge 0$ and $\sum_{u \in L_t} c_u > 0$, the correlation satisfies
\[
\frac{\Cov(\zeta,X_0)}{\sqrt{\Var(\zeta)}} \ll_d \frac{1}{\sqrt[4]{t}}.
\]
\end{theorem}
We prove Theorem~\ref{thm:conv-corr} in Section~\ref{s:Poisson}.

\subsubsection{The Infinite Model}
Next, we consider the \defn{infinite model} in which the poset is the ``half-space''
\[
P = \HS_{d+1} := \left\{ (x_1, \ldots, x_{d+1}) \in \mathbb{Z}^{d+1} \, | \, x_1 + \cdots + x_{d+1} \ge 0 \right\},
\]
with $u \ge v$ if and only if $u - v \in \mathbb{Z}_{\ge 0}^{d+1}$. For each $t \ge 0$, the $t^{\text{th}}$ layer of this poset is
\[
L_t = \left\{ (x_1, \ldots, x_{d+1}) \in \mathbb{Z}^{d+1} \, | \, x_1 + \cdots + x_{d+1} = t \right\}.
\]
Note that every non-minimal element in $P$ covers exactly $d+1$ elements, and thus the only parameter we need to consider is $\alpha_{d+1}$. For simplicity, let us write $\alpha := \alpha_{d+1}$.

When $d = 0$, the poset $P = \HS_1$ is the same as $\mathbb{Z}_{\ge 0}$, which was considered earlier in the finite model.

When $d \ge 1$, our model places the same random variable $X_0$ at every minimal element $u$ in the infinite set $L_0$. Therefore, it is not hard to see that reconstruction and convex reconstruction are always possible regardless of the value of $\alpha$, since on every layer $L_t$ there are infinitely many identically distributed Gaussian random variables with the same correlation with $X_0$ and with independent noises. A more interesting question is whether {\em local} reconstructions are possible.

\begin{theorem}\label{thm:half-space}
Let $d \ge 0$ and $P = \HS_{d+1}$.
\begin{itemize}
    \item[(a)] If $\alpha > \frac{1}{d+1}$, then single-vertex reconstruction is possible.
    \item[(b)] If $\alpha < \frac{1}{d+1}$, then local reconstruction is not possible.
    \item[(c)] If $\alpha = \frac{1}{d+1}$, then single-vertex reconstruction is possible for $d \ge 3$ and local reconstruction is not possible for $d = 0, 1, 2$.
\end{itemize}
\end{theorem}

In particular, we have the following corollary.

\begin{corollary}
For $P = \HS_{d+1}$, local reconstruction is possible if and only if single-vertex reconstruction is possible.
\end{corollary}

We prove Theorem~\ref{thm:half-space} in Section~\ref{s:half-space}.

\subsubsection{Summary}

Table~\ref{table:reconstruction} provides a summary of our results about the parameters for which each type of reconstruction is possible.

Note that by definition of reconstruction, to show that any type of reconstruction is possible, it suffices to prove that the limit superior $\limsup_{t \to \infty} \sup_{\zeta} \frac{\Cov(\zeta,X_0)}{\sqrt{\Var(\zeta)}}$ is positive (where the supremum $\sup_{\zeta}$ is over a family of estimators $\zeta$ depending on which type of reconstruction we are considering---see Remark~\ref{rem:different-S-t}). 
We remark that in all of our positive reconstruction results in this paper, our proofs in fact show the stronger results that the {\em limit inferior} is positive.

\begin{table}
\begin{center}
\def\arraystretch{1.5}
{\tabulinesep=0.2em
\begin{tabu}{|c|c|c|}
\hline
& the infinite model & the finite model \\
\hline
reconstruction 
& \makecell[l]{\multicolumn{1}{l}{($d = 0$) $\alpha > 1$} \\[0.5em] \multicolumn{1}{l}{($d \ge 1$) any $\alpha > 0$} \\[0.5em] \hspace{0.2em} (easy)}
& \makecell[l]{$\exists i \in [d+1], \alpha_i > \frac{1}{i}$ \\[0.5em] (Thm.~\ref{thm:rec-orthant})} \\
\hline
convex reconstruction 
& \makecell[l]{\multicolumn{1}{l}{($d = 0$) $\alpha > 1$} \\[0.5em] \multicolumn{1}{l}{($d \ge 1$) any $\alpha > 0$} \\[0.5em] \hspace{0.2em} (easy)} 
& \makecell[l]{$\exists i \in [d+1], \alpha_i > \frac{1}{i}$ \\[0.5em] (Thm.~\ref{thm:rec-orthant})} \\ 
\hline
local reconstruction 
& \hspace{0.2em} \makecell[l]{($d \le 2$) $\alpha > \frac{1}{d+1}$ \\[0.5em] ($d \ge 3$) $\alpha \ge \frac{1}{d+1}$ \\[0.5em] (Thm.~\ref{thm:half-space})}
& \makecell[c]{open \\[0.5em] (see \S\S~\ref{subsec:local-finite})} \\
\hline
local convex reconstruction 
& \hspace{0.2em} \makecell[l]{($d \le 2$) $\alpha > \frac{1}{d+1}$ \\[0.5em] ($d \ge 3$) $\alpha \ge \frac{1}{d+1}$ \\[0.5em] (Thm.~\ref{thm:half-space})}
& \makecell[c]{open \\[0.5em] (see \S\S~\ref{subsec:local-finite})} \\
\hline
single-vertex reconstruction 
& \hspace{0.2em} \makecell[l]{($d \le 2$) $\alpha > \frac{1}{d+1}$ \\[0.5em] ($d \ge 3$) $\alpha \ge \frac{1}{d+1}$ \\[0.5em] (Thm.~\ref{thm:half-space})}
& \makecell[c]{open \\[0.5em] (see \S\S~\ref{subsec:local-finite})} \\
\hline
\end{tabu}
}
\end{center}

\bigskip

\caption{When is each type of reconstruction possible?}\label{table:reconstruction}
\end{table}

\section{Parameter Reduction}\label{s:para-reduct}
In this section, we show that whether each type of reconstruction is possible does not depend on the parameters $\varepsilon, \mu_0, \sigma_0^2$.

\begin{proposition}\label{prop:para-reduct}
Reconstruction is possible for $P$ and $(\alpha_1, \ldots, \alpha_{d+1}, \varepsilon, \mu_0, \sigma_0^2)$ if and only if it is for the same poset and the tuple $(\alpha_1, \ldots, \alpha_{d+1}, 1, 0, 1)$.
\end{proposition}
\begin{proof}
Let us begin by taking $\Cov(X_0, -)$ throughout the model recurrence \eqref{eq:model-rec}. We find that every $v \in P$ with rank at least $1$ satisfies
\[
\Cov(X_v, X_0) = \alpha_{|\fp(v)|} \cdot \sum_{u \in \fp(v)} \Cov(X_u, X_0).
\]
It then follows by induction on the rank that for every $v \in P$, the covariance $\Cov(X_v, X_0)$ equals the sum over all saturated chains (for poset terminologies, see Subsubsection~\ref{subsubsec:poset-terms}) from any minimal element $u \in L_0$ to $v$ of the weight of the chain, where the weight of the chain $\{ u = z_0 \lessdot z_1 \lessdot \cdots \lessdot z_r = v\}$ is given by the monomial
\[
\alpha_{|\fp(z_1)|} \cdot \alpha_{|\fp(z_2)|} \cdots \alpha_{|\fp(z_r)|}.
\]

Hence, for every $u \in P$, we can express $X_u$ uniquely as
\[
X_u = F_u \cdot X_0 + Y_u,
\]
where $F_u \in \mathbb{Z}[\alpha_1, \ldots, \alpha_{d+1}]$ is a polynomial depending only on $\alpha_1, \ldots, \alpha_{d+1}$, and $Y_u$ is a combination of the noises $W_{w \to w'}$ ($w \lessdot w' \le u$) independent from $X_0$. By a similar argument, for each pair $u, v \in P$, there exists a polynomial $G_{uv} \in \mathbb{Z}[\alpha_1, \ldots, \alpha_{d+1}]$ depending only on $\alpha_1, \ldots, \alpha_{d+1}$ such that
\[
\Cov(Y_u, Y_v) = \sum_{w \lessdot w'} \Cov(Y_u, W_{w \to w'}) \cdot \Cov(Y_v, W_{w \to w'}) = \varepsilon^2 \cdot G_{uv}.
\]
To see how we obtain the factor of $\varepsilon^2$ on the right, note that each of $\Cov(Y_u, W_{w \to w'})$ and $\Cov(Y_v, W_{w \to w'})$ equals $\varepsilon$ times a polynomial in $\mathbb{Z}[\alpha_1, \ldots, \alpha_{d+1}]$---another consequence of the model recurrence \eqref{eq:model-rec}.

For each $t \ge 0$, let us consider
\begin{equation}\label{eq:S-t-sup}
S_t := \sup \frac{\Cov(\zeta, X_0)}{\sqrt{\Var(\zeta)}},
\end{equation}
where the supremum is over all choices of
\[
\zeta = \sum_{u \in L_t} c_u X_u
\]
such that (i) the coefficients $c_u$ are not simultaneously zero, and (ii) all but finitely many coefficients $c_u$ are zero.

By writing $\Cov(\zeta, X_0)$ and $\Var(\zeta)$ explicitly in terms of $F_u$ and $G_{uv}$, we find that
\begin{equation}\label{eq:s4e-2}
\sigma_0^4 \varepsilon^{-2} S_t^{-2} - \sigma_0^2 \varepsilon^{-2} = \inf \frac{\sum_{u,v \in L_t} c_u c_v G_{uv}}{\left( \sum_{u \in L_t} c_u F_u \right)^2},
\end{equation}
where the infimum is over all choices of $c_u \in \mathbb{R}$ satisfying the conditions (i) and (ii) in the previous paragraph.

Reconstruction is not possible for the tuple $(\alpha_1, \ldots, \alpha_{d+1}, \varepsilon, \mu_0, \sigma_0^2)$ if and only if $S_t \to 0$ (as $t \to \infty$), which is if and only if the expression in \eqref{eq:s4e-2} diverges to $+\infty$. Note that the right-hand side of \eqref{eq:s4e-2} does not depend on $\varepsilon$, $\mu_0$, or $\sigma_0^2$. This concludes the proof.
\end{proof}

\begin{remark}\label{rem:different-S-t}
Similarly, we can show that whether convex/local/local convex/single-vertex reconstruction is possible does not depend on $\varepsilon, \mu_0, \sigma_0^2$. To see this, we simply use a slight modification of the above proof. Instead of taking the supremum in \eqref{eq:S-t-sup} and the infimum in \eqref{eq:s4e-2} over all choices of $c_u \in \mathbb{R}$ satisfying (i) and (ii) in the proof, we take the supremum and the infimum over a different family of $\{c_u\}_{u \in L_t}$ according to the definition of reconstruction we are considering. The rest of the proof reads the same.
\end{remark}

In light of Proposition~\ref{prop:para-reduct} and Remark~\ref{rem:different-S-t}, let us assume $\varepsilon = 1$, $\mu_0 = 0$, and $\sigma_0^2 = 1$ for the rest of the paper. It now makes sense to say whether reconstruction is possible for the poset $P$ and the parameter tuple $(\alpha_1, \ldots, \alpha_{d+1})$, where we drop the three parameters from the tuple.

\section{The Finite Model}\label{s:orthant}
In this section, we consider the finite model, where $P = \mathbb{Z}_{\ge 0}^{d+1}$. Our goal is to establish Theorem~\ref{thm:rec-orthant}, which gives the characterization of the parameter tuples for which (convex) reconstruction is possible.

In the simplest case when $d = 0$, the theorem says that reconstruction is possible if and only if $\alpha_1 > 1$. This result is easy to show as follows. The poset in this case is simply $P = \mathbb{Z}_{\ge 0}$. For each $t$, we have
\[
X_t = \alpha_1^t X_0 + \sum_{i=0}^{t-1} \alpha_1^{t-i} W_{i, i+1},
\]
and therefore
\begin{equation}\label{eq:var-cov-2-d-0}
\frac{\Var(X_t)}{\Cov(X_t,X_0)^2} = 1 + \sum_{i=0}^{t-1} \alpha_1^{-2i}.
\end{equation}
When $\alpha_1 > 1$, the quantity in \eqref{eq:var-cov-2-d-0} is bounded from above by the convergent sum $1+ \sum_{i=0}^{\infty} \alpha_1^{-2i} = (2\alpha_1^2 - 1)/(\alpha_1^2 - 1)$. On the other hand, when $\alpha_1 \le 1$, the quantity diverges to $+\infty$ as $t \to \infty$.

For the rest of this section, let us consider $d \ge 1$. We prove parts~(a) and (b) of Theorem~\ref{thm:rec-orthant} in Subsections~\ref{subsec:in-box} and \ref{subsec:out-box}, respectively.

\subsection{The Box of Impossible Reconstruction}\label{subsec:in-box}
In this subsection, we fix a positive integer $d \ge 1$ and consider real numbers $\alpha_1, \ldots, \alpha_{d+1}$ such that $\alpha_i \in (0,1/i]$ for every $i \in [d+1]$. Consider arbitrary real numbers $a_u$ indexed by the poset elements $u \in L_t$ on the $t^{\text{th}}$ layer. We assume that these numbers $a_u$ are not simultaneously zero. Consider
\begin{equation}\label{eq:zeta-sum-u-L-t}
\zeta := \sum_{u \in L_t} a_u X_u.
\end{equation}
The goal of this subsection is to show Theorem~\ref{thm:corr-log-log-log-log-t}, which implies that
\[
\frac{|\Cov(\zeta, X_0)|}{\sqrt{\Var(\zeta)}} \to 0,
\]
as $t \to \infty$.

It is convenient to consider the directed graph obtained from the Hasse diagram of the poset $P$ by drawing a directed edge $u \to v$ whenever $u \in \fp(v)$ (i.e. $v$ covers $u$). To each directed edge $u \to v$, we assign the \defn{weight} given by
\[
\wt(u \to v) := \alpha_{|\fp(v)|}.
\]
More generally, for any two poset elements $u, v \in P$, if $\gamma$ is a directed path from $u$ to $v$ in the graph, the \defn{weight} $\wt(\gamma)$ is defined to be the product of the constituent edges of $\gamma$.

For each function $f$ and a nonnegative integer $k$, we use the notation $f^{(k)}$ for the composition $f \circ f \circ \cdots \circ f$ with $k$ copies of $f$. If $k = 0$, then $f^{(0)}$ denotes the identity function.

For the rest of this subsection:
\begin{itemize}
    \item We fix the level $t$ and assume that $t$ is large with respect to $d$. For convenience, we take $t > \exp^{(5d+4)}(1)$.
    \item We fix the coefficients $a_u$, for $u \in L_t$, of the random variable $\zeta$. See~\eqref{eq:zeta-sum-u-L-t}.
\end{itemize}

Now we extend the real numbers $a_u$ to other layers. For each $0 \le m \le t$ and $v \in L_m$, define
\[
a_v := \sum_{u \in L_t} \sum_{\gamma: v \to u} \wt(\gamma) \cdot a_u,
\]
where the inner sum is over all directed paths $\gamma$ from $v$ to $u$. Note that
\[
\Cov(\zeta, X_0) = \sum_{u \in L_t} a_u \cdot \Cov(X_u,X_0) = \sum_{u \in L_t} a_u \sum_{\gamma: \zv \to u} \wt(\gamma) = a_{\zv}.
\]

Define the constant $C = C(\alpha_1, \ldots, \alpha_{d+1}) := \max_{i \in [d+1]} i^{-1/2} \alpha_i^{-1} \ge \sqrt{d+1} > 0$.

\begin{proposition}\label{prop:1-C-2-Var}
We have
\begin{itemize}
    \item[(a)] $\Var(\zeta) = a_{\zv}^2 + \sum_{m=1}^t \sum_{w \in L_m} |\fp(w)| \alpha_{|\fp(w)|}^2 a_w^2$.
    \item[(b)] $\frac{1}{C^2} \sum_{m = 0}^t \sum_{v \in L_m} a_v^2 \le \Var(\zeta) \le \sum_{m = 0}^t \sum_{v \in L_m} a_v^2$.
\end{itemize}
\end{proposition}
\begin{proof}
{\bf (a)} From the model recurrence \eqref{eq:model-rec}, we can write for each $u \in L_t$,
\[
X_u = \sum_{\gamma: \zv \to u} \wt(\gamma) X_0 + \sum_{\substack{v, w \\ v \lessdot w \le u}} \alpha_{|\fp(w)|} \sum_{\gamma: w \to u} \wt(\gamma) W_{v \to w}.
\]
Therefore,
\begin{align*}
\zeta &= \sum_{u \in L_t} a_u \sum_{\gamma: \zv \to u} \wt(\gamma) X_0 + \sum_{u \in L_t} a_u \sum_{\substack{v, w \\ v \lessdot w \le u}} \alpha_{|\fp(w)|} \sum_{\gamma: w \to u} \wt(\gamma) W_{v \to w} \\
&= a_{\zv} X_0 + \sum_{\substack{v, w \\ v \lessdot w}} \alpha_{|\fp(w)|} W_{v\to w} \sum_{u \in L_t} \sum_{\gamma: w \to u} \wt(\gamma) a_u \\
&= a_{\zv} X_0 + \sum_{m=1}^t \sum_{w \in L_m} \sum_{v \in \fp(w)} \alpha_{|\fp(w)|} a_w W_{v \to w}.
\end{align*}
Hence,
\[
\Var(\zeta) = a_{\zv}^2 + \sum_{m=1}^t \sum_{w \in L_m} |\fp(w)| \alpha_{|\fp(w)|}^2 a_w^2.
\]

\medskip

\noindent {\bf (b)} This part follows from using part~{\bf (a)} with the observation that for each summand above,
\[
\frac{1}{C^2} \le |\fp(w)| \alpha_{|\fp(w)|}^2 \le 1,
\]
by the definition of the constant $C$.
\end{proof}

Let us now introduce some useful notations. For each $j \in [d+1]$, we define
\[
\lambda_j := \left\lfloor \log^{(2d+2-2j)}(t) \right\rfloor.
\]
Let us decompose
\[
L_{\lambda_j} = S_j^1 \cup S_j^2 \cup \cdots \cup S_j^{d+1},
\]
where $S_j^i$ is the set of all $(d+1)$-tuples $u \in L_{\lambda_j}$ with exactly $i$ positive entries. We also use the convenient notation
\[
S^{\ge i}_j := \bigcup_{i' \ge i} S^{i'}_j.
\]

For each pair $u, v \in P$, define
\begin{equation}\label{eq:p-u-to-v-def}
p(u \to v) := \sum_{\gamma:u \to v} \wt(\gamma).
\end{equation}
For each $m \ge 0$ and $v = (v_1, \ldots, v_{d+1}) \in L_m$, we use the notation $\binom{m}{v}$ for the multinomial coefficient
\[
\binom{m}{v} = \binom{m}{v_1, \ldots, v_{d+1}} := \frac{m!}{v_1! \cdots v_{d+1}!}.
\]

\begin{proposition}\label{prop:i-j-1-p-u-v}
For any $i, j_1, j_2 \in [d+1]$ with $j_1 < j_2$, and for any $v \in S^i_{j_2}$, we have
\[
\sum_{u \in S^i_{j_1}} p(u \to v) \le \binom{\lambda_{j_2}}{v} i^{\lambda_{j_1} - \lambda_{j_2}}.
\]
\end{proposition}
\begin{proof}
Let $\tau \in L_i$ be the $\{0,1\}$-vector where each entry is $1$ if and only if the corresponding entry in $v$ is positive. We have
\begin{equation}\label{eq:p-tau-v-ge}
p(\tau \to v) = \sum_{\substack{u \in S^i_{j_1} \\ u \ge \tau}} p(\tau \to u) p(u \to v) \ge \sum_{\substack{u \in S^i_{j_1} \\ u \ge \tau}} \alpha_i^{\lambda_{j_1} - i} p(u \to v).
\end{equation}
On the other hand,
\begin{equation}\label{eq:p-tau-v-le}
p(\tau \to v) = \alpha_i^{\lambda_{j_2} - i} \binom{\lambda_{j_2} - i}{v - \tau} \le \alpha_i^{\lambda_{j_2} - i} \binom{\lambda_{j_2}}{v}.
\end{equation}
Combining \eqref{eq:p-tau-v-ge} and \eqref{eq:p-tau-v-le}, we find
\[
\sum_{u \in S^i_{j_1}} p(u \to v) = \sum_{\substack{u \in S^i_{j_1} \\ u \ge \tau}} p(u \to v) \le \alpha_i^{\lambda_{j_2} - \lambda_{j_1}} \binom{\lambda_{j_2}}{v} \le \binom{\lambda_{j_2}}{v} i^{\lambda_{j_1} - \lambda_{j_2}},
\]
as desired.
\end{proof}

\begin{proposition}\label{prop:sum-p-le-1}
For any $v \in P$ and for any $t \ge 0$, we have
\[
\sum_{u \in L_t} p(u \to v) \le 1.
\]
\end{proposition}
\begin{proof}
Suppose that $v \in L_{t'}$. If $t > t'$, then the sum is zero. If $t = t'$, then the sum is one. For $t < t'$, we argue by induction on $t' - t$ as follows.
\begin{align*}
\sum_{u \in L_t} p(u \to v) &= \sum_{u \in L_t} \sum_{w \in \fp(v)} p(u \to w) p(w \to v) \\
&= \sum_{w \in \fp(v)} \alpha_{|\fp(v)|} \sum_{u \in L_t} p(u \to w) \\
&\le |\fp(v)| \alpha_{|\fp(v)|} \le 1,
\end{align*}
where the first inequality follows from the inductive hypothesis, and the second from the parameter assumption.
\end{proof}

\begin{proposition}\label{prop:L1-monotone}
For any $i, j_1, j_2 \in [d+1]$ with $j_1 < j_2$, we have
\[
\sum_{u \in S^{\ge i}_{j_1}} |a_u| \le \sum_{u \in S^{\ge i}_{j_2}} |a_u|.
\]
\end{proposition}
\begin{proof}
We have
\begin{align*}
\sum_{u \in S^{\ge i}_{j_1}} |a_u| &\le \sum_{u \in S^{\ge i}_{j_1}} \sum_{v \in S^{\ge i}_{j_2}} p(u \to v) |a_v| \\
&= \sum_{v \in S^{\ge i}_{j_2}} |a_v| \sum_{u \in S^{\ge i}_{j_1}} p(u \to v) \\
&\le \sum_{v \in S^{\ge i}_{j_2}} |a_v|,
\end{align*}
where we have used the triangle inequality and Proposition~\ref{prop:sum-p-le-1}.
\end{proof}

\begin{lemma}\label{lem:i-2-l-j-i}
For any $i,j \in [d+1]$, we have
\[
\sum_{v \in S^i_j} \binom{\lambda_j}{v}^2 \le 2^{d+1} \frac{i^{2\lambda_j + i}}{\lambda_j^{(i-1)/2}}.
\]
\end{lemma}
\begin{proof}
Each element $v \in S^i_j$ has exactly $i$ positive entries and $d+1-i$ zero entries. This decomposes $S^i_j$ into $\binom{d+1}{i}$ subsets of equal size. Using this decomposition, we find
\[
\sum_{v \in S^i_j} \binom{\lambda_j}{v}^2 = \binom{d+1}{i} \sum_{\substack{v_1, \ldots, v_i > 0 \\ v_1 + \cdots + v_i = \lambda_j}} \binom{\lambda_j}{v_1, \ldots, v_i}^2 \le 2^{d+1} \frac{i^{2\lambda_j + i}}{\lambda_j^{(i-1)/2}},
\]
where the final step uses Lemma~\ref{l:general-i} with $\lambda_j$ and $i$ playing the roles of $n$ and $i$ in the lemma statement. Note that the assumption in the lemma is satisfied, since $\lambda_j \ge \exp^{(3d+3)}(1)$ and $d+1 \ge i$.
\end{proof}

For each $j \in [d+1]$, let us define $x_j$ by
\[
\sum_{u \in S_j^{\ge j+1}} |a_u| = x_j \cdot C \sqrt{\Var(\zeta)}.
\]
Note that $x_{d+1} = 0$.

Using Proposition~\ref{prop:L1-monotone}, we find, for each $j \in [d]$,
\begin{align}
\sum_{u \in S_j^{\ge j+1}} |a_u| &= \sum_{u \in S_j^{j+1}} |a_u| + \sum_{u \in S_j^{\ge j+2}} |a_u| \label{eq:u-in-S-j-j-1} \\
&\le \sum_{u \in S_j^{j+1}} |a_u| + \sum_{u \in S_{j+1}^{\ge j+2}} |a_u| \notag \\
&= \sum_{u \in S_j^{j+1}} |a_u| + x_{j+1} \cdot C \sqrt{\Var(\zeta)}. \label{eq:ge-j1-j-1}
\end{align}

Let us analyze the sum in \eqref{eq:ge-j1-j-1}. By the triangle inequality, we have
\begin{align}
\sum_{u \in S_j^{j+1}} |a_u| &\le \sum_{u \in S_j^{j+1}} \sum_{v \in S_{j+1}^{\ge j+1}} p(u \to v) |a_v| \notag \\
&= \sum_{v \in S^{\ge j+2}_{j+1}} |a_v| \sum_{u \in S^{j+1}_j} p(u \to v) + \sum_{v \in S^{j+1}_{j+1}} |a_v| \sum_{u \in S^{j+1}_j} p(u \to v). \label{eq:j-2-j-1}
\end{align}
The first sum in \eqref{eq:j-2-j-1} can be bounded using Proposition~\ref{prop:sum-p-le-1}:
\begin{equation}
\sum_{v \in S^{\ge j+2}_{j+1}} |a_v| \sum_{u \in S^{j+1}_j} p(u \to v) \le \sum_{v \in S^{\ge j+2}_{j+1}} |a_v| = x_{j+1} \cdot C \sqrt{\Var(\zeta)}. \label{eq:v-in-S-j2-j1}
\end{equation}
For the second sum in \eqref{eq:j-2-j-1}, we have
\begin{align}
\sum_{v \in S^{j+1}_{j+1}} |a_v| \sum_{u \in S^{j+1}_j} p(u \to v)
&\le \sum_{v \in S_{j+1}^{j+1}} |a_v| \binom{\lambda_{j+1}}{v} (j+1)^{\lambda_j - \lambda_{j+1}} \label{eq:v-j1-j1} \\
&\le (j+1)^{\lambda_j - \lambda_{j+1}} \left( \sum_{v \in S^{j+1}_{j+1}} a_v^2 \right)^{1/2} \left( \sum_{v \in S^{j+1}_{j+1}} \binom{\lambda_{j+1}}{v}^2 \right)^{1/2} \notag \\
&\le (j+1)^{\lambda_j - \lambda_{j+1}} \cdot C \sqrt{\Var(\zeta)} \cdot \sqrt{2^{d+1} \frac{(j+1)^{2\lambda_{j+1} + j + 1}}{\lambda_{j+1}^{j/2}}}, \notag
\end{align}
where we have used Proposition~\ref{prop:i-j-1-p-u-v}, Cauchy--Schwarz, and Lemma~\ref{lem:i-2-l-j-i}, respectively.

Combining \eqref{eq:u-in-S-j-j-1}, \eqref{eq:j-2-j-1}, \eqref{eq:v-in-S-j2-j1}, and \eqref{eq:v-j1-j1}, we obtain
\[
x_j \le 2x_{j+1} + \frac{(j+1)^{\lambda_j} 2^{(d+1)/2} (j+1)^{(j+1)/2}}{\lambda_{j+1}^{j/4}}.
\]
Hence,
\[
2^j x_j \le 2^{j+1} x_{j+1} + \frac{(j+1)^{\lambda_j}}{\lambda_{j+1}^{j/4}} \cdot 2^{3d} d^d.
\]
Now we use the telescoping trick to obtain
\[
2x_1 \le 2^{3d} d^d \sum_{j=1}^d \frac{(j+1)^{\lambda_j}}{\lambda_{j+1}^{j/4}}.
\]
From the parameter assumptions, for every $j \in [d]$, we have
\[
\frac{(j+1)^{\lambda_j}}{\lambda_{j+1}^{j/4}} \le \frac{1}{\sqrt[5]{\log^{(2d-1)}(t)}}.
\]
Therefore,
\begin{equation}\label{eq:log-2d-t}
\sum_{u \in S_1^{\ge 2}} |a_u| = x_1 \cdot C \sqrt{\Var(\zeta)} \le \frac{C \sqrt{\Var(\zeta)}}{\log^{(2d)}(t)}.
\end{equation}

Now for each positive integer $m$ such that $1 \le m \le \log^{(2d+1)}(t)$, let us write
\[
L_m = L'_m \uplus L''_m,
\]
where $L'_m$ is the $(d+1)$-element subset of $L_m$ containing the tuples with only one positive entry. Using a triangle inequality argument similar to the proof of Proposition~\ref{prop:L1-monotone}, we obtain
\begin{equation}\label{ineq:u-in-L''-m}
\sum_{u \in L''_m} |a_u| \le \sum_{u \in S^{\ge 2}_1} |a_u| \overset{\eqref{eq:log-2d-t}}{\le} \frac{C \sqrt{\Var(\zeta)}}{\log^{(2d)}(t)}.
\end{equation}

We can now prove the main theorem of this subsection.
\begin{theorem}\label{thm:corr-log-log-log-log-t}
We have
\[
\frac{|\Cov(\zeta, X_0)|}{\sqrt{\Var(\zeta)}} \le \frac{C}{\log^{(2d+2)}(t)}.
\]
\end{theorem}
\begin{proof}
We consider two cases.

\medskip

\underline{Case 1.} Suppose that
\[
|a_{\zv}| \le \frac{C\sqrt{\Var(\zeta)}}{\log^{(2d+1)}(t)}.
\]
In this case,
\[
\frac{|\Cov(\zeta, X_0)|}{\sqrt{\Var(\zeta)}} = \frac{|a_{\zv}|}{\sqrt{\Var(\zeta)}} \le \frac{C}{\log^{(2d+1)}(t)} \le \frac{C}{\log^{(2d+2)}(t)}.
\]

\medskip

\underline{Case 2.} Suppose that
\begin{equation}\label{ineq:a-zv-log-2d-1}
|a_{\zv}| > \frac{C\sqrt{\Var(\zeta)}}{\log^{(2d+1)}(t)}.
\end{equation}
For each positive integer $m \le \log^{(2d+1)}(t)$, we have by the triangle inequality,
\begin{align*}
|a_{\zv}| &\le \sum_{u \in L'_m} p(\zv \to u) |a_u| + \sum_{u \in L''_m} p(\zv \to u) |a_u| \\
&\overset{\eqref{ineq:u-in-L''-m}}{\le} \sum_{u \in L'_m} |a_u| + \frac{C\sqrt{\Var(\zeta)}}{\log^{(2d)}(t)} \\
&\overset{\eqref{ineq:a-zv-log-2d-1}}{\le} \sum_{u \in L'_m} |a_u| + \frac{1}{2} | a_{\zv} |.
\end{align*}
Hence, by Cauchy--Schwarz,
\[
\frac{1}{2}|a_{\zv}| \le \sum_{u \in L'_m} |a_u| \le \sqrt{d+1} \sqrt{\sum_{u \in L'_m} a_u^2 },
\]
which gives
\[
\sum_{u \in L'_m} a_u^2  \ge \frac{a_{\zv}^2}{4(d+1)}.
\]
Using Proposition~\ref{prop:1-C-2-Var}, we find
\[
C^2 \Var(\zeta) \ge \sum_{m \le \log^{(2d+1)}(t)} \, \sum_{u \in L'_m} a_u^2 \ge \frac{a_{\zv}^2}{4(d+1)} \left\lfloor \log^{(2d+1)}(t) \right\rfloor,
\]
which gives
\[
\frac{|\Cov(\zeta, X_0)|}{\sqrt{\Var(\zeta)}} \le \frac{2C\sqrt{d+1}}{\sqrt{\left\lfloor \log^{(2d+1)}(t) \right\rfloor}} \le \frac{C}{\log^{(2d+2)}(t)},
\]
as desired.
\end{proof}

\subsection{Outside the Box}\label{subsec:out-box}
In this subsection, we prove part~(b) of Theorem~\ref{thm:rec-orthant}.
\begin{lemma}\label{l:d-d'}
For any nonnegative integers $d < d'$ and for any positive real numbers $\alpha_1, \ldots, \alpha_{d'+1}$, if reconstruction is possible for the poset $\mathbb{Z}_{\ge 0}^{d+1}$ and the tuple $(\alpha_1, \ldots, \alpha_{d+1})$, then reconstruction is possible for the poset $\mathbb{Z}_{\ge 0}^{d'+1}$ and the tuple $(\alpha_1, \ldots, \alpha_{d'+1})$.
\end{lemma}
\begin{proof}
If reconstruction is possible for the poset $\mathbb{Z}_{\ge 0}^{d+1}$ and the tuple $(\alpha_1, \ldots, \alpha_{d+1})$, then there exists a sequence $\{\zeta_t\}_{t=0}^{\infty}$ of random variables
\[
\zeta_t = \sum_{\substack{u = (u_1, \ldots, u_{d+1}) \in \mathbb{Z}_{\ge 0}^{d+1} \\ u_1 + \cdots + u_{d+1} = t}} c_u X_u
\]
such that
\[
\frac{\Cov\!\left(\zeta_t, X_0 \right)}{\sqrt{\Var\!\left(\zeta_t\right)}} \nrightarrow 0,
\]
as $t \to \infty$. In the poset $\mathbb{Z}_{\ge 0}^{d'+1}$, we simply construct
\[
\zeta'_t := \sum_{\substack{u = (u_1, \ldots, u_{d'+1}) \in \mathbb{Z}_{\ge 0}^{d'+1} \\ u_1 + \cdots + u_{d'+1} = t}} c'_u X_u,
\]
whose coefficients are given by
\[
c'_{(u_1, \ldots, u_{d'+1})} := \begin{cases}
c_{(u_1, \ldots, u_{d+1})} & \text{if } u_{d+2} = \cdots = u_{d'+1} = 0, \\
0 & \text{otherwise}.
\end{cases}
\]
It follows from the construction and the model that
\[
\frac{\Cov\!\left(\zeta'_t, X\right)}{\sqrt{\Var\!\left(\zeta'_t\right)}} = \frac{\Cov\!\left(\zeta_t, X\right)}{\sqrt{\Var\!\left(\zeta_t\right)}},
\]
and therefore reconstruction is possible for the poset $\mathbb{Z}_{\ge 0}^{d'+1}$ and the tuple $(\alpha_1, \ldots, \alpha_{d'+1})$
\end{proof}

Now we show that if there exists an index $i \in [d+1]$ such that $\alpha_i > 1/i$, then reconstruction is possible. By Lemma~\ref{l:d-d'}, we may assume that $\alpha_{d+1} > 1/(d+1)$.

For each $t \ge 1$, define
\[
\calZ_t := \sum_{u \in L_t} \kappa_{|\fp(u)|} X_u,
\]
where for each $i \in [d+1]$,
\begin{equation}\label{eq:kappa-i}
\kappa_i := (d+1-i)! \left( \prod_{1 \le j \le i-1} ((d+1)\alpha_{d+1} - j \alpha_j) \right) \left( \prod_{i+1 \le j \le d+1} \alpha_j \right).
\end{equation}
Using the model recurrence \eqref{eq:model-rec}, we see that for each $t \ge 1$,
\begin{align}
\calZ_{t+1} &= \sum_{u \in L_{t+1}} \kappa_{|\fp(u)|} X_u \notag \\
&= \sum_{u \in L_{t+1}} \kappa_{|\fp(u)|}\alpha_{|\fp(u)|} \sum_{v \lessdot u} (X_v + W_{v \to u}) \notag \\
&= \sum_{v \in L_t} \sum_{u \gtrdot v} \kappa_{|\fp(u)|} \alpha_{|\fp(u)|} X_v + \xi_t, \label{eq:k-p-u-a-p-u}
\end{align}
where $\xi_t$ is a linear combination of the ``noises'' $W_{v \to u}$ such that
\[
\Var(\xi_t) \ll_{\alpha_1, \ldots, \alpha_{d+1}} t^d.
\]
The coefficient of $X_v$ in the expression on the right-hand side of~\eqref{eq:k-p-u-a-p-u} is
\begin{align*}
\sum_{u \gtrdot v} \kappa_{|\fp(u)|} \alpha_{|\fp(u)|} &= |\fp(v)| \kappa_{|\fp(v)|} \alpha_{|\fp(v)|} + (d+1-|\fp(v)|) \cdot \kappa_{|\fp(v)|+1} \alpha_{|\fp(v)|+1} \\
&= (d+1) \alpha_{d+1} \cdot \kappa_{|\fp(v)|},
\end{align*}
where in the final step we used the definition of $\kappa_i$ from \eqref{eq:kappa-i}. Hence,
\[
\calZ_{t+1} = (d+1)\alpha_{d+1} \cdot \calZ_t + \xi_t,
\]

Now we take
\[
\zeta_t := \frac{\calZ_t}{((d+1)\alpha_{d+1})^t}.
\]
We have $\Cov(\zeta_t, X_0) = \Cov(\zeta_1, X_0)$, and
\[
\Var(\zeta_{t+1}) = \Var(\zeta_t) + \frac{\Var(\xi_t)}{((d+1)\alpha_{d+1})^{2t+2}}.
\]
Since $(d+1)\alpha_{d+1} > 1$, we find that there exists a finite number $D > 0$ such that for every $t$, we have $\Var(\zeta_t) < D$. This shows that
\[
\frac{\Cov(\zeta_t, X_0)}{\sqrt{\Var(\zeta_t)}} \ge \frac{\Cov(\zeta_1, X_0)}{\sqrt{D}} > 0,
\]
and thus reconstruction is possible.

\subsection{Rate of Convergence}\label{subsec:exact-formula}
In the proof of Proposition~\ref{prop:para-reduct}, we have defined the useful quantity
\[
S_t := \sup \frac{\Cov(\zeta, X_0)}{\sqrt{\Var(\zeta)}}.
\]
Reconstruction is not possible if and only if $S_t \to 0$ as $t \to \infty$.

We have seen that reconstruction is not possible when $\alpha_i \le 1/i$ for every $i \in [d+1]$. Equivalently, $S_t \to 0$ as $t \to \infty$. It is interesting to consider the rate at which $S_t$ decays to $0$. In particular, let us consider the \defn{critical case} where $\alpha_i = 1/i$ for every $i \in [d+1]$. Our argument in Subsection~\ref{subsec:in-box} shows that
\[
S_t \ll \frac{1}{\log^{(2d+2)}(t)}.
\]
We believe that the upper bound for the rate of convergence above is far from sharp.

In this subsection, we consider the critical case when $d = 1$. We show that in this case the quantity $S_t$ can be computed {\em exactly} and we give a formula for it in closed form. For each $t \in \bbZ_{\ge 0}$, we define the following bivariate polynomial
\[
Q_t(z,w) := \sum_{i=0}^t \sum_{j=0}^t \Cov\!\left( X_{(i,t-i)}, X_{(j,t-j)} \right) \, z^i w^j \in \bbR[z,w].
\]

Below we use the following standard notation for extracting a ``coefficient'' out of a polynomial. If
\[
R(z,w) = \sum_{i,j} r_{i,j} z^i w^j,
\]
then we write
\[
[w^t](R(z,w)) = \sum_i r_{i,t} z^i.
\]

\begin{lemma}\label{l:S_t-A_t}
Suppose that there exist $c_0, c_1, \ldots, c_t \in \bbR$ and $A_t \in \bbR \setminus \{0\}$ such that we have the following equality of polynomials
\[
[w^t]\left(\left( c_t + c_{t-1} w + \cdots + c_0 w^t \right) \cdot Q_t(z,w)\right) = A_t \cdot \sum_{i=0}^t z^i \in \bbR[z].
\]
Then
\[
S_t = \sqrt{\frac{\sum_{i=0}^t c_i}{A_t}}.
\]
\end{lemma}
\begin{proof}
Let us write
\[
\zeta := \sum_{j=0}^t c_j X_{(j,t-j)}.
\]
The equality of polynomials implies that for every $i$,
\[
\Cov\!\left( X_{(i,t-i)}, \zeta \right) = A_t,
\]
and thus
\[
\Cov\!\left( X_{(i,t-i)}, X_0 - \frac{\zeta}{A_t} \right) = 0.
\]
This means that
\[
\frac{\zeta}{A_t} = \bbE\!\left[ X_0 \, | \, X_{(0,t)}, X_{(1,t-1)}, \ldots, X_{(t,0)} \right].
\]
Recall the standard argument that the conditional expectation above is a maximizer for
\[
\frac{\Cov\!\left( - , X_0 \right)}{\sqrt{\Var( - )}}.
\]
Hence, $S_t = \frac{\Cov(\zeta/A_t, X_0)}{\sqrt{\Var(\zeta/A_t)}} = \frac{\Cov(\zeta,X_0)}{\sqrt{\Var(\zeta)}}$. Now, since
\[
\Cov\!\left( \zeta, X_0 \right) = \Cov\!\left( \sum_{j=0}^t c_j X_{(j,t-j)}, X_0 \right) = \sum_{j=0}^t c_j,
\]
and
\[
\Var\!\left( \zeta \right) = \sum_{j=0}^t c_j \Cov\!\left( X_{(j,t-j)}, \zeta \right) = A_t \cdot \sum_{j=0}^t c_j,
\]
we conclude that
\[
S_t = \sqrt{\frac{\sum_{i=0}^t c_i}{A_t}},
\]
as desired.
\end{proof}

The polynomial $Q_t(z,w)$ can be computed explicitly. Using standard generating function techniques, the model recurrence \eqref{eq:model-rec} also gives a recurrence for $Q_t$. The functional equation relating $Q_{t+1}$ and $Q_t$ can be solved explicitly. We omit the standard tedious details and give the result here.
For each $t \in \bbZ_{\ge 0}$, we have
\begin{align}
Q_t(z,w) &= - \frac{2zw}{(1-z)(1-w)^2}\left(\frac{z(1+w)}{2}\right)^t - \frac{2zw}{(1-z)^2(1-w)}\left(\frac{(1+z)w}{2}\right)^t \label{eq:Q_t-explicit} \\
&\hphantom{=} +\frac{(t+1)(zw)^{t+1}}{(1-z)(1-w)} +\frac{2-6zw+2z^2w+2zw^2}{(1-z)^2(1-w)^2(1-zw)}(zw)^{t+1} \notag \\
&\hphantom{=} +\frac{2z}{(1-z)^2(1-w)}\left(\frac{1+z}{2}\right)^t+\frac{2w}{(1-z)(1-w)^2}\left(\frac{1+w}{2}\right)^t \notag \\
&\hphantom{=} - \frac{z^{t+1}}{(1-z)(1-w)} - \frac{w^{t+1}}{(1-z)(1-w)} + \frac{t+1}{(1-z)(1-w)} \notag \\
&\hphantom{=} + \frac{-2z-2w+6zw-2z^2w^2}{(1-z)^2(1-w)^2(1-zw)}. \notag
\end{align}

Using the formula in Equation~(\ref{eq:Q_t-explicit}), we can check the following equality directly:
\[
[w^t]\left((1+w)^t\cdot Q_t(z,w)\right) = A_t \sum_{i=0}^t z^i,
\]
where
\[
A_t = \frac{t}{2^t} \binom{2t}{t} + 2^t.
\]
Using Lemma~\ref{l:S_t-A_t} (with $c_i = \binom{t}{i}$), we conclude that
\[
S_t = \frac{1}{\sqrt{\frac{t}{4^t}\binom{2t}{t} + 1}},
\]
as desired. Note that the asymptotic formula
\[
S_t \sim \pi^{1/4} t^{-1/4},
\]
as $t \to \infty$, follows easily.

\begin{remark}\label{rem:exact-zeta}
Our argument also shows that the linear combination
\[
\widehat{\zeta} := \sum_{i=0}^t \binom{t}{i} X_{(i,t-i)}
\]
is a maximizer of
\[
\zeta \mapsto \frac{\Cov(\zeta, X_0)}{\sqrt{\Var(\zeta)}},
\]
among all nonzero linear combinations of random variables on $L_t$, for every $t \in \bbZ_{\ge 0}$. One might guess that for the general dimension $d$ the optimal linear combination on the $t^{\text{th}}$ layer of $\mathbb{Z}_{\ge 0}^{d+1}$ would be
\[
\zeta = \sum_{u \in L_t} \binom{t}{u} X_u,
\]
where the coefficients are multinomial coefficients. Unfortunately, we know this guess is not correct even for when $d = 2$. Explicit computations on low layers $L_t$ show that some coefficients in an optimal linear combination can be {\em negative} (while others in the same linear combination are positive). It remains an interesting question whether there is a ``nice'' formula for an optimal linear combination in the critical case in any dimension $d \ge 2$.
\end{remark}

\section{Convex Reconstruction and Poisson Clocks}\label{s:Poisson}
In this section, we fix the poset $P$ to be $\mathbb{Z}_{\ge 0}^{d+1}$, the same one we considered in the previous section, and consider the critical case, when $\alpha_i = 1/i$ for every $i \in [d+1]$. We take an arbitrary convex combination
\[
\zeta = \sum_{u \in L_t} c_u X_u,
\]
where $c_u \ge 0$ with $\sum_{u \in L_t} c_u > 0$. Our goal of this subsection is to show that
\begin{equation}\label{eq:ll-d-sqrt-4-t}
\frac{\Cov(\zeta,X_0)}{\sqrt{\Var(\zeta)}} \ll_d \frac{1}{\sqrt[4]{t}}.
\end{equation}
To that end, we may make the following assumptions.
\begin{itemize}
    \item By rescaling, we may assume
    \begin{equation}\label{eq:sum-c-u-1}
    \sum_{u \in L_t} c_u = 1.
    \end{equation}
    \item By symmetry of the coordinates, we may assume
    \begin{equation}\label{eq:1-d-1-factorial}
        \sum_{u \in L^{\nearrow}_t} c_u \ge \frac{1}{(d+1)!},
    \end{equation}
    where $L^{\nearrow}_t$ denotes the set of all $(d+1)$-tuples $(u_1, \ldots, u_{d+1}) \in L_t$ such that $u_1 \le \cdots \le u_{d+1}$.
\end{itemize}

Note that since we are in the critical case, \eqref{eq:sum-c-u-1} implies that
\[
\Cov(\zeta,X_0) = 1,
\]
and hence we would like to bound $\Var(\zeta)$ from below.

Now we consider $d+1$ independent Poisson clocks ${\tt C}_1, \ldots, {\tt C}_{d+1}$, each with rate $1$. We construct a random walk $\Gamma$ on $P$ as follows. At time $0$, let $\Gamma(0)$ be $u = (u_1, \ldots, u_{d+1}) \in L_t$ with probability $c_u$. At each time any of the clocks rings---say ${\tt C}_i$ rings---we take a step
\[
(x_1, \ldots, x_i, \ldots, x_{d+1}) \mapsto \begin{cases}
    (x_1, \ldots, x_i - 1, \ldots, x_{d+1}), & \text{ if } x_i > 0, \\
    (x_1, \ldots, 0, \ldots, x_{d+1}), & \text{ otherwise.}
\end{cases}
\]

With probability $1$, the random walk $\Gamma$ forms a chain of elements of $P$ from some vertex on the $t^{\text{th}}$ layer to $\zv \in L_0$ in finite time. Since we are in the critical case with $\alpha_i = 1/i$ for every $i \in [d+1]$, this random chain has the same probability law as the chain obtained by starting on the $t^{\text{th}}$ layer according to $\{c_u \, : \, u \in L_t\}$ and in each step moving from a vertex $u \in P$ to a vertex in $\fp(u)$ uniformly at random, independently from all previous steps, until we hit $\zv$. Thus it is not hard to see that for any $w \in P$ and $u \in L_t$, conditioning on $\Gamma(0) = u$, the probability that $w$ is in this chain is exactly
\begin{equation}\label{eq:cond-prob-w-to-u}
p(w \to u) = \sum_{\gamma: w \to u} \wt(\gamma).
\end{equation}
\begin{proposition}\label{prop:w-in-Gamma-squared}
We have
\[
\Var(\zeta) \ge \frac{1}{d+1} \sum_{w \in P} \bbP\!\left( w \in \Gamma \right)^2.
\]
\end{proposition}
\begin{proof}
From \eqref{eq:cond-prob-w-to-u}, we find that
\begin{equation}\label{eq:P-w-in-Gamma}
\bbP( w \in \Gamma) = \sum_{u \in L_t} \bbP( w \in \Gamma \, | \, \Gamma(0) = u) \cdot \bbP(\Gamma(0) = u) = \sum_{u \in L_t} c_u p(w \to u).
\end{equation}

Using the formula for the variance of $\zeta$ from Proposition~\ref{prop:1-C-2-Var}(a), we find
\begin{align*}
\Var(\zeta) &= 1 + \sum_{v \lessdot w} \left( \sum_{u \in L_t} c_u p(w \to u) \frac{1}{|\fp(w)|} \right)^2 \\
&= 1 + \sum_{w \in P - \{\zv\}} \frac{1}{|\fp(w)|} \cdot \left( \sum_{u \in L_t} c_u p(w \to u) \right)^2 \\
&\overset{\eqref{eq:P-w-in-Gamma}}{\ge} \frac{1}{d+1} \sum_{w \in P} \bbP( w \in \Gamma )^2,
\end{align*}
as desired.
\end{proof}

For each $u \in L^\nearrow_t$, we let $\calE_u$ denote the event that at time $u_{d+1}$,
\begin{itemize}
    \item for each $i \in [d]$, the clock ${\tt C}_i$ has rung at least $u_{d+1}$ times, and
    \item the clock ${\tt C}_{d+1}$ has rung at most $u_{d+1} - \sqrt{u_{d+1}}$ times.
\end{itemize}

To analyze this event $\calE_u$, let us prove the following lemma about Poisson random variables.

\begin{lemma}
Let $T$ be a positive integer. Let $Z \sim \Pois(T)$. Then
\begin{itemize}
    \item[(a)] $\bbP(Z \ge T) \ge 1/2$.
    \item[(b)] $\bbP(Z \le T - \sqrt{T}) \ge e^{-9}$.
\end{itemize}
\end{lemma}
\begin{proof}
\noindent \textbf{(a)} Write $p_k := \bbP(Z = k) = e^{-T} \frac{T^k}{k!}$. Note that for each $k = 0, 1, \ldots, T-1$, we have
\[
\frac{p_k}{p_{2T-1-k}} = \prod_{i=1}^{T-1-k} \left( 1 - \frac{k^2}{T^2} \right) \le 1,
\]
and thus $p_k \le p_{2T-1-k}$. This shows
\[
p_0 + \cdots + p_{T-1} \le p_T + \cdots + p_{2T-1}.
\]
Hence,
\[
\bbP(Z \ge T) \ge p_T + \cdots + p_{2T-1} \ge \frac{1}{2},
\]
as the sum $p_0 + p_1 + \cdots$ equals $1$.

\medskip

\noindent \textbf{(b)} This is clear when $T \le 9$. For the rest of this proof, assume $T \ge 10$. Write $\lambda := \left\lceil \sqrt{T} \right\rceil$. Note that the finite sequence $p_0, p_1, \ldots, p_{T-\lambda}$ is increasing. Thus
\[
\bbP(Z \le T- \sqrt{T}) \ge \lambda \cdot p_{T-2\lambda} \ge \sqrt{T} e^{-T} \frac{T^{T-2\lambda}}{(T-2\lambda)!} \ge e^{-9}.
\]
The last inequality holds for $T \ge 10$, and can be shown using (exact inequality versions of) Stirling's approximation.
\end{proof}

Since the Poisson clocks ${\tt C}_1, \ldots, {\tt C}_{d+1}$ are independent, we obtain the following corollary.

\begin{corollary}\label{cor:2-d-e-9}
For each $u \in L^\nearrow_t$, we have $\bbP(\calE_u) \ge 2^{-d} e^{-9}$.
\end{corollary}

\begin{proposition}\label{prop:w-in-Gamma-2-d-e-9}
For each $w = (0,\ldots, 0, w_{d+1}) \in P$ with $w_{d+1} \le \sqrt{\frac{t}{d+1}}$, and for each $u \in L^\nearrow_t$, we have
\[
\bbP(w \in \Gamma \, | \, \Gamma(0) = u) \ge 2^{-d} e^{-9}.
\]
\end{proposition}
\begin{proof}
Observe that
\[
\bbP(w \in \Gamma \, | \, \Gamma(0) = u) \ge \bbP(w \in \Gamma \, | \, \Gamma(0) = u, \calE_u) \cdot \bbP(\calE_u).
\]
Under the events $\calE_u$ and $\Gamma(0) = u$, at time $u_{d+1}$, the random walk $\Gamma$ is at
\[
\Gamma(u_{d+1}) = (0,\ldots, 0, x),
\]
for a certain integer $x \ge \sqrt{u_{d+1}} \ge \sqrt{\frac{t}{d+1}}$. Afterwards, the random walk can only go in straight line to $\zv$. This means that $\Gamma$ has to pass $w$ with probability $1$. Combining this with Corollary~\ref{cor:2-d-e-9}, we obtain the desired bound.
\end{proof}

Using Proposition~\ref{prop:w-in-Gamma-2-d-e-9} with \eqref{eq:1-d-1-factorial}, we find that for any $w$ in the assumption of the proposition,
\[
\bbP(w \in \Gamma) \ge \sum_{u \in L^\nearrow_t} \bbP(w \in \Gamma \, | \, \Gamma(0) = u) \cdot \bbP(\Gamma(0) = u) \ge \frac{2^{-d}e^{-9}}{(d+1)!}.
\]
With Proposition~\ref{prop:w-in-Gamma-squared}, we conclude
\[
\Var(\zeta) \ge \frac{1}{d+1} \sum_{\substack{w = (0,\ldots, 0,w_{d+1}) \\ w_{d+1} \le \sqrt{t/(d+1)}}} \bbP(w \in \Gamma)^2 \ge \frac{4^{-d} e^{-18}}{(d+1)^{3/2} (d+1)!^2} \cdot \sqrt{t},
\]
from which the estimate \eqref{eq:ll-d-sqrt-4-t} follows.

\section{The Infinite Model}\label{s:half-space}
In this section, we consider the infinite model, where $P = \HS_{d+1}$. The goal of this section is to prove Theorem~\ref{thm:half-space}. We prove part~(a) in Subsection~\ref{subsec:supercritical}, part~(b) in Subsection~\ref{subsec:subcritical}, and part~(c) in Subsection~\ref{subsec:critical} below.

\subsection{Covariance and Variance Computations}
For each vector $v \in \mathbb{Z}^{d+1}$, let $|v|$ denote the sum of its entries. Our key technical tool is the following estimate on the sum of squares of multinomial coefficients.
\begin{lemma}\label{l:C-1-C-2}
For each $d \ge 0$, there exist constants $C_1(d), C_2(d) > 0$ depending only on $d$ such that for every $k \in \mathbb{Z}_{\ge 1}$, we have
\[
C_1(d) \cdot \frac{(d+1)^{2k}}{k^{d/2}} \le \sum_{\substack{v \in \mathbb{Z}_{\ge 0}^{d+1} \\ |v| = k}} \binom{k}{v}^2 \le C_2(d) \cdot \frac{(d+1)^{2k}}{k^{d/2}}.
\]
\end{lemma}
\begin{proof}
See e.g. \cite{RS09}.
\end{proof}

Let $t \ge 1$. For each $v \in L_t$, we can express
\[
X_v = \left((d+1)\alpha_{d+1}\right)^t X_0 + \sum_{w' \lessdot w} \alpha_{d+1} p(w \to v) W_{w' \to w}.
\]
Therefore,
\begin{equation}\label{eq:Cov-X-u-X-0}
\Cov(X_v,X_0) = ((d+1)\alpha_{d+1})^t,
\end{equation}
and
\begin{equation}\label{eq:var-x-u}
\Var(X_v) = ((d+1)\alpha_{d+1})^{2t} + \sum_{w' \lessdot w} \alpha_{d+1}^2 p(w \to v)^2.
\end{equation}
More generally, for any $v, v' \in L_t$, we have
\[
\Cov(X_v, X_{v'}) = ((d+1)\alpha_{d+1})^{2t} + \sum_{w' \lessdot w} \alpha_{d+1}^2 p(w \to v) p(w \to v').
\]
By splitting the summation according to which layer $w$ is in, we obtain
\[
\Cov(X_v, X_{v'}) = ((d+1)\alpha_{d+1})^{2t} + \sum_{k=0}^{t-1} (d+1)\alpha_{d+1}^2 \sum_{\substack{w \in L_{t-k} \\ w \le v \\ w \le v'}} \alpha_{d+1}^k \binom{k}{v-w} \cdot \alpha_{d+1}^k \binom{k}{v'-w}.
\]
By making a change of variables $w \mapsto v' - w$ and using the convention that $\binom{k}{w_1, \ldots, w_{d+1}} = 0$ whenever some entry $w_i$ is negative, we can write the formula above as
\begin{equation}\label{eq:cov-formula}
\Cov(X_v, X_{v'}) = ((d+1)\alpha_{d+1})^{2t} + (d+1)\alpha_{d+1}^2 \sum_{k=0}^{t-1} \alpha_{d+1}^{2k} \sum_{w \in \mathbb{Z}^{d+1}} \binom{k}{w} \binom{k}{w+v-v'}.
\end{equation}

\subsection{Supercritical Case}\label{subsec:supercritical}
In this subsection, assume $(d+1)\alpha_{d+1} > 1$. To show that single-vertex reconstruction is possible, we simply take $\zeta_t = X_u$ for any $u \in L_t$.

Using Lemma~\ref{l:C-1-C-2}, \eqref{eq:Cov-X-u-X-0}, and \eqref{eq:var-x-u}, we find that for each $t \ge 1$,
\begin{align*}
\frac{\Var(X_u)}{\Cov(X_u,X_0)^2} &\le 1 + (d+1)^{1-2t}\alpha_{d+1}^{2-2t} + (d+1)\alpha_{d+1}^2 C_2(d) \sum_{k=1}^{t-1} \frac{((d+1)\alpha_{d+1})^{2k-2t}}{k^{d/2}} \\
&\le 1 + (d+1)\alpha_{d+1}^2 + C_2(d) \frac{(d+1)^3 \alpha_{d+1}^4}{(d+1)^2\alpha_{d+1}^2 - 1}.
\end{align*}

\subsection{Subcritical Case}\label{subsec:subcritical}
In this subsection, we assume $(d+1)\alpha_{d+1} < 1$. Consider an arbitrary nonzero
\[
\zeta = \sum_{u \in L_t \cap \calW_t} a_u X_u,
\]
where $\calW_t$ is a window of width $N$. The covariance formula \eqref{eq:Cov-X-u-X-0} shows that
\[
\Cov(\zeta,X_0) = ((d+1)\alpha_{d+1})^t \cdot \sum_{u \in L_t \cap \calW_t} a_u.
\]
Next, from the variance formula \eqref{eq:var-x-u}, using the noises between the $(t-1)^{\text{st}}$ and the $t^{\text{th}}$ layers, we have the lower bound
\[
\Var(\zeta) \ge (d+1)\alpha_{d+1}^2 \cdot \sum_{u \in L_t \cap \calW_t} a_u^2.
\]
Since there are at most $N^d$ vertices in $L_t \cap \calW_t$, we obtain by Cauchy--Schwarz that
\[
\frac{\Cov(\zeta,X_0)}{\sqrt{\Var(\zeta)}} 
\le \frac{((d+1)\alpha_{d+1})^t \cdot \sum_{u \in L_t \cap \calW_t} a_u}{\alpha_{d+1} \sqrt{d+1} \cdot \sqrt{\sum_{u \in L_t \cap \calW_t} a_u^2}} 
\le \frac{N^{d/2}}{\alpha_{d+1} \sqrt{d+1}} \cdot ((d+1)\alpha_{d+1})^t.
\]
With $(d+1)\alpha_{d+1} < 1$, this shows that the correlation decays exponentially with the layer, and thus local reconstruction is not possible.

\subsection{Critical Case}\label{subsec:critical}
In this subsection, assume $(d+1)\alpha_{d+1} = 1$. In this critical case, for every $u \in L_t$, we have $\Cov(X_u, X_0) = 1$, and
\begin{equation}\label{eq:var-formula-crit}
\Var(X_u) = 1 + \frac{1}{d+1} + \frac{1}{d+1} \sum_{k=1}^{t-1} \frac{1}{(d+1)^{2k}} \sum_{\substack{v \in \mathbb{Z}_{\ge 0}^{d+1} \\ |v| = k}} \binom{k}{v}^2.
\end{equation}
Hence, when $d \ge 3$, we obtain the bound
\[
\Var(X_u) \le 2 + \frac{1}{4} C_2(d) {\boldsymbol{\zeta}}(d/2) < \infty,
\]
where the notation ${\boldsymbol{\zeta}}$ in the bold font denotes the Riemann zeta function. This shows that single-vertex reconstruction is possible for $d \ge 3$.

Now we turn to the critical case when $d \le 2$. We claim that local reconstruction is not possible. First, when $d = 0$, the poset $P$ is the same as $\mathbb{Z}_{\ge 0}$, where it is easy to see that reconstruction is not possible.

When $d = 1$, we prove the following result.

\begin{proposition}\label{prop:window-d-1}
For any positive integer $N$, and for any real numbers $a_1, \ldots, a_N$ (which are not simultaneously zero), the random variable
\[
\zeta := a_1 X_{(1,t-1)} + \cdots + a_N X_{(N,t-N)}
\]
satisfies
\[
\frac{\Cov(\zeta,X_0)}{\sqrt{\Var(\zeta)}} \le \frac{2^{14} N^{3/2}}{\sqrt[4]{t}}.
\]
\end{proposition}

Before giving the proof of the proposition, we make a quick remark about the strategy. We use an ``averaging trick'' where we consider a new random variable $\widehat{\zeta}$ whose correlation with $X_0$ is at least as good as that of $\zeta$ (see \eqref{ineq:better-corr} below). After averaging, the proof reduces to understanding the behavior of the covariance of nearby random variables on the same layer, which we analyze in Lemma~\ref{l:sqrt-t-1-50-sqrt-t} in the next section.

\begin{proof}[Proof of Proposition~\ref{prop:window-d-1}]
Write
\[
\varepsilon := \frac{\Cov(\zeta,X_0)}{\sqrt{\Var(\zeta)}}.
\]
We consider two cases.

\underline{Case 1.} Suppose $\varepsilon \le 2^{14} N^{3/2}t^{-1/3}$. Then
\[
\varepsilon \le \frac{2^{14} N^{3/2}}{\sqrt[3]{t}} \le \frac{2^{14} N^{3/2}}{\sqrt[4]{t}}
\]
follows immediately.

\underline{Case 2.} Now suppose $\varepsilon > 2^{14} N^{3/2}t^{-1/3}$. In this case, we take
\[
M := \left\lceil \frac{2^9 N^{3/2}}{\varepsilon} \right\rceil.
\]
Note that we now have positive integers $M$ and $t$ such that
\begin{equation}\label{ineq:M-N-t}
M \ge \frac{2^9 N^{3/2}}{\varepsilon} \qquad \text{and} \qquad t \ge 2^{12} M^3.
\end{equation}

Without loss of generality, let us assume that $\sum_{i=1}^N a_i^2 = 1$. Write
\[
A := \Cov(\zeta, X_0) = a_1 + \cdots + a_N.
\]
By considering the noises between the last two layers ($W_{w' \to w}$ with $w' \in L_{t-1}$ and $w \in L_t$), we obtain a lower bound on the variance:
\[
\Var(\zeta) \ge \frac{1}{2}\left( a_1^2 + \cdots + a_N^2 \right) = \frac{1}{2},
\]
which implies
\begin{equation}\label{ineq:A-eps-sqrt-2}
A \ge \frac{\varepsilon}{\sqrt{2}}.
\end{equation}
We define a new random variable by averaging:
\begin{align*}
&\widehat{\zeta} := \frac{1}{M} \sum_{i = 1}^M \left( a_1 X_{(i,t-i)} + a_2 X_{(i+1,t-i-1)} + \cdots + a_N X_{(i+N-1,t-i-N+1)} \right) \\
&= \frac{(a_1 X_{(1,t-1)} + \cdots + a_N X_{(N,t-N)}) + \cdots + (a_1 X_{(M,t-M)} + \cdots + a_N X_{(M+N-1, t-M-N+1)})}{M}.
\end{align*}
Observe that $\Cov(\widehat{\zeta},X_0) = \Cov(\zeta,X_0)$ while $\Var(\widehat{\zeta}) \le \Var(\zeta)$. Consequently,
\begin{equation}\label{ineq:better-corr}
\frac{\Cov(\widehat{\zeta},X_0)}{\sqrt{\Var(\widehat{\zeta})}} \ge \frac{\Cov(\zeta,X_0)}{\sqrt{\Var(\zeta)}} = \varepsilon.
\end{equation}
Let us write $\widehat{\zeta}$ as
\[
\widehat{\zeta} = b_1 X_{(1,t-1)} + b_2 X_{(2,t-2)} + \cdots + b_{M+N-1} X_{(M+N-1,t-M-N+1)},
\]
where
\[
b_i = \begin{cases}
\frac{a_1 + \cdots + a_i}{M} & \text{ if } i \le N - 1, \\
\frac{A}{M} & \text{ if } N \le i \le M, \text{ and } \\
\frac{a_{i-M+1} + \cdots + a_N}{M} & \text{ if } i \ge M + 1.
\end{cases}
\]
From the explicit formula above, we obtain from Cauchy--Schwarz that
\begin{equation}\label{ineq:upper-bound-on-b-i}
|b_i| \le \frac{|a_1| + \cdots + |a_N|}{M} \le \frac{1}{M} \sqrt{a_1^2 + \cdots + a_N^2} \sqrt{1^2 + \cdots + 1^2} = \frac{\sqrt{N}}{M},
\end{equation}
for every $i \in [M+N-1]$. Now let
\[
E := \sum_{i \le N-1} b_i X_{(i,t-i)} + \sum_{i \ge M+1} b_i X_{(i,t-i)},
\]
so that $\widehat{\zeta} = \frac{A}{M}(X_N + \cdots + X_M) + E$.

Using Lemma~\ref{l:sqrt-t-1-50-sqrt-t}, which we prove below in Section~\ref{s:cov-est}, and the bound \eqref{ineq:upper-bound-on-b-i}, we find
\begin{align*}
\Var(\widehat{\zeta}) &= \Var\!\left( \frac{A}{M}(X_N + \cdots + X_M) + E\right) \\
&\ge \frac{A^2}{M^2} \Var(X_N + \cdots + X_M) + \frac{2A}{M} \Cov\!\left( X_N + \cdots + X_M, E \right) \\
&\ge \frac{A^2}{M^2} (M-N+1)^2 \cdot \frac{\sqrt{t}}{50} - \frac{2A}{M} \cdot (M-N+1) \cdot 2(N-1)\frac{\sqrt{N}}{M} \cdot \sqrt{t}.
\end{align*}
Therefore,
\begin{equation}\label{ineq:AM-2}
\frac{\Var(\widehat{\zeta})}{\Cov(\widehat{\zeta},X_0)^2} \ge \left( \frac{1}{50} \left( \frac{M-N+1}{M} \right)^2 - \frac{4(N-1)(M-N+1)\sqrt{N}}{AM^2} \right) \sqrt{t}.
\end{equation}
From \eqref{ineq:better-corr}, the left-hand side of \eqref{ineq:AM-2} is bounded from above by $1/\varepsilon^2$. From \eqref{ineq:M-N-t}, we find that
\[
\frac{1}{50} \left( \frac{M-N+1}{M} \right)^2 > 0.019,
\]
and
\[
\frac{4(N-1)(M-N+1)\sqrt{N}}{AM^2} \overset{\eqref{ineq:A-eps-sqrt-2}}{\le} \frac{4\sqrt{2}N^{3/2}}{M\varepsilon} < 0.012.
\]
Combining these estimates yields
\[
\frac{1}{\varepsilon^2} \ge 0.007 \sqrt{t} > \frac{\sqrt{t}}{2^8},
\]
from which it follows that
\[
\varepsilon < \frac{2^4}{\sqrt[4]{t}} < \frac{2^{14} N^{3/2}}{\sqrt[4]{t}},
\]
as desired.
\end{proof}

When $d = 2$, we prove the following result. The proof follows the same strategy as the one for the previous proposition.

\begin{proposition}[cf.~Prop.~\ref{prop:window-d-1}]
Let $t \ge 2$ be an integer. For any positive integer $N$, and for any nonzero real matrix $A = [a_{(i,j)}]_{i,j \in [N]} \neq \boldsymbol{0}_{N \times N}$, the random variable
\[
\zeta := \sum_{(i,j) \in [N]^2} a_{(i,j)} X_{(i,j,t-i-j)}
\]
satisfies
\[
\frac{\Cov(\zeta,X_0)}{\sqrt{\Var(\zeta)}} \le \frac{2^{24} N^2}{\sqrt{\log t}}.
\]
\end{proposition}
\begin{proof}
This proof proceeds in a similar manner as that of Proposition~\ref{prop:window-d-1} above. Let $\varepsilon$ denote the correlation. We consider two cases.

\underline{Case 1.} Suppose $\varepsilon \le 2^{21} N^2 t^{-1/6}$. Then
\[
\varepsilon \le \frac{2^{21} N^2}{\sqrt[6]{t}} \le \frac{2^{24} N^2}{\sqrt{\log t}}.
\]

\underline{Case 2.} Now suppose $\varepsilon > 2^{21} N^2 t^{-1/6}$. Take
\begin{equation}\label{eq:2-16-N-2}
M := \left\lceil \frac{2^{16}N^2}{\varepsilon} \right\rceil.
\end{equation}
Note that we have $t \ge 2^{24} M^6$.

Like before, we assume without loss of generality that $\sum_{(i,j) \in [N]^2} a_{(i,j)}^2 = 1$, and thus
\begin{equation}\label{eq:A-sqrt-3}
A := \Cov(\zeta, X_0) = \sum_{(i,j) \in [N]^2} a_{(i,j)} \ge \frac{\varepsilon}{\sqrt{3}}.
\end{equation}

Using the averaging trick again, we define
\[
\widehat{\zeta} := \frac{1}{M^2} \sum_{(\Delta_1, \Delta_2) \in [M]^2} \sum_{(i,j) \in [N]^2} a_{(i,j)} X_{(i+\Delta_1, j+\Delta_2, t-i-j-\Delta_1-\Delta_2)}.
\]
Performing a similar computation as before and invoking Lemma~\ref{l:log-t-400}, we discover that
\[
\frac{\Var(\widehat{\zeta})}{\Cov(\widehat{\zeta},X_0)^2} \ge \left( \frac{1}{400} \left( \frac{M-N+1}{M} \right)^4 - \frac{8N^2}{AM} \right) \log t.
\]
Hence, using \eqref{eq:2-16-N-2} and \eqref{eq:A-sqrt-3}, we obtain
\[
\frac{1}{\varepsilon^2} \ge \frac{\log t}{2^{10}},
\]
which shows
\[
\varepsilon \le \frac{2^5}{\sqrt{\log t}} \le \frac{2^{24} N^2}{\sqrt{\log t}},
\]
as desired.
\end{proof}

\section{Covariance Estimates in the Infinite Model}\label{s:cov-est}
Here we prove some technical lemmas for $\HS_2$ and $\HS_3$.

\subsection{Binomials}
In this subsection, we consider $P = \HS_2$. We prove the following lemma.

\begin{lemma}\label{l:sqrt-t-1-50-sqrt-t}
Let $M$ and $t$ be positive integers such that $M \ge 2^5$ and $t \ge 2^{11} M^3$. Then for integers $i,j$ with $|i-j| \le 2M$, we have
\[
\sqrt{t} \ge \Var\!\left( X_{(i,t-i)} \right) \ge \Cov\!\left( X_{(i,t-i)}, X_{(j,t-j)} \right) \ge \frac{1}{50} \sqrt{t}.
\]
\end{lemma}

\begin{proof}
Let us prove the inequalities above from the left to the right.

For the first inequality, the variance formula \eqref{eq:var-formula-crit} implies
\[
\Var\!\left( X_{(i,t-i)} \right) = \frac{3}{2} + \frac{1}{2} \sum_{k=1}^{t-1} 4^{-k} \binom{2k}{k}.
\]
We use the estimate
\[
4^{-k} \binom{2k}{k} \le \frac{1}{\sqrt{\pi k}},
\]
which holds for every positive integer $k$, and which can be proved for example by an exact form of Stirling's approximation. Thus,
\[
\Var\!\left( X_{(i,t-i)} \right) \le \frac{3}{2} + \frac{1}{2} \sum_{k=1}^{t-1} \frac{1}{\sqrt{\pi k}} \le \frac{3}{2} + \sqrt{\frac{t-1}{\pi}} \le \sqrt{t}.
\]

\smallskip

The second inequality in the lemma follows simply from
\[
\Cov\!\left( X_{(i,t-i)}, X_{(j,t-j)} \right) \le \sqrt{\Var\!\left( X_{(i,t-i)} \right) \cdot \Var\!\left( X_{(j,t-j)} \right)} = \Var\!\left( X_{(i,t-i)} \right).
\]

\smallskip

Now we turn to the third inequality. For convenience, let us write $u := |i-j| \le 2M$. From the covariance formula \eqref{eq:cov-formula}, we have a lower bound
\begin{equation}\label{eq:cov-from-t-2-to-t-1}
\Cov\!\left( X_{(i,t-i)}, X_{(j,t-j)} \right) \ge \frac{1}{2} \sum_{k = \left\lceil t/2 \right\rceil}^{t-1} 4^{-k} \sum_{w \in \mathbb{Z}} \binom{k}{w} \binom{k}{w+u}.
\end{equation}
Here, we use the convention that $\binom{k}{w}$ is zero whenever $w < 0$ or $w > k$.

We are going to bound the inner sum in the right-hand side of \eqref{eq:cov-from-t-2-to-t-1} from below. Note that since the outer sum takes $k$ from $\left\lceil t/2 \right\rceil$ to $t-1$, we obtain from the parameter assumption that $k \ge t/2 \ge 2^{10} M^3$, and in particular that $k \ge 2^{25}$, which is helpful as we perform various exact bounds below. Observe that
\begin{align}
\sum_{w \in \mathbb{Z}} \binom{k}{w}\binom{k}{w+u}
&\ge
\sum_{\substack{w \in \mathbb{Z} \\ \left| w - \frac{k}{2} \right| \le k^{2/3}}} \binom{k}{w}^2 \cdot \frac{(k-w) \cdots (k-w-u+1)}{(w+1) \cdots (w+u)} \notag \\
&\ge
\sum_{\substack{w \in \mathbb{Z} \\ \left| w - \frac{k}{2} \right| \le k^{2/3}}} \binom{k}{w}^2 \left( \frac{k-w-u+1}{w+u} \right)^2 \notag \\
&\ge
\left( 1 - \frac{5}{\sqrt[3]{k}} \right)^{2M} \sum_{\substack{w \in \mathbb{Z} \\ \left| w - \frac{k}{2} \right| \le k^{2/3}}} \binom{k}{w}^2. \label{ineq:k-w-u-1}
\end{align}
The tail sum can be upper bounded by the magnitude of each term as
\[
\sum_{\left| w - \frac{k}{2} \right| > k^{2/3}} \binom{k}{w}^2 \le (k+1) \binom{k}{\left\lceil \frac{k}{2} - k^{2/3} \right\rceil}^2.
\]
Using an exact version of Stirling's approximation, we obtain the estimate
\[
\binom{k}{\left\lceil \frac{k}{2} - k^{2/3} \right\rceil} \le \frac{2^{k+6}}{\sqrt{k}} e^{-2k^{1/3}}.
\]
Combining the two inequalities above with the parameter assumption on $k$, we obtain
\begin{equation}\label{ineq:k-w-u-2}
\sum_{\left| w - \frac{k}{2} \right| > k^{2/3}} \binom{k}{w}^2 \le \frac{1}{2^{1000}} \cdot \frac{4^k}{k^2}.
\end{equation}
Note that the huge denominator $2^{1000}$ follows because $k$ is assumed to be large.

Now observe that
\begin{equation}\label{ineq:k-w-u-3}
\sum_{w \in \mathbb{Z}} \binom{k}{w}^2 = \binom{2k}{k} \ge \frac{4^k}{2\sqrt{k}},
\end{equation}
and
\begin{equation}\label{ineq:k-w-u-4}
\left( 1 - \frac{5}{\sqrt[3]{k}} \right)^{2M} \ge \frac{1}{4},
\end{equation}
from the assumption $k \ge 2^{10} M^3$.

Combining \eqref{ineq:k-w-u-1}--\eqref{ineq:k-w-u-4}, we obtain
\[
\sum_{w \in \mathbb{Z}} \binom{k}{w} \binom{k}{w+u} \ge \frac{1}{4} \left( \frac{1}{2} \cdot \frac{4^k}{\sqrt{k}} - \frac{1}{2^{1000}} \cdot \frac{4^k}{k^2} \right) > \frac{1}{10} \cdot \frac{4^k}{\sqrt{k}}.
\]

Using this estimate in \eqref{eq:cov-from-t-2-to-t-1}, we conclude that
\[
\Cov\!\left( X_{(i,t-i)}, X_{(j,t-j)} \right) \ge \frac{1}{20} \sum_{k = \left\lceil t/2 \right\rceil}^{t-1} \frac{1}{\sqrt{k}} \ge \frac{1}{50}\sqrt{t},
\]
as desired.
\end{proof}

\subsection{Trinomials}
In this subsection, we consider the poset $P = \HS_3$. We prove the following lemma.

\begin{lemma}[cf.~Lem.~\ref{l:sqrt-t-1-50-sqrt-t}]\label{l:log-t-400}
Let $M$ and $t$ be integers such that $M \ge 2^5$ and $t \ge 2^{24} M^6$. Then for integers $i, j, i', j'$ with $\max\{|i-i'|, |j-j'|\} \le 2M$, we have
\[
\log t \ge \Var\!\left( X_{(i,j,t-i-j)} \right) \ge \Cov\!\left( X_{(i,j,t-i-j)}, X_{(i',j',t-i'-j')} \right) \ge \frac{\log t}{400}.
\]
\end{lemma}
\begin{proof}
We proceed in the same manner as our proof of Lemma~\ref{l:sqrt-t-1-50-sqrt-t}.

Let us prove the inequalities in the lemma from the left to the right. Using the variance formula \eqref{eq:var-formula-crit}, we find
\begin{align*}
\Var\!\left( X_{(i,j,t-i-j)} \right) 
&= \frac{4}{3} + \frac{1}{3} \sum_{k=1}^{t-1} 3^{-2k} \sum_{\substack{v \in \mathbb{Z}_{\ge 0}^3 \\ |v| = k}} \binom{k}{v}^2 \\
&\overset{(\text{Lem.}~\ref{lem:sqrt-27-4-pi})}{\le} \frac{4}{3} + \frac{1}{3} \sum_{k=1}^{t-1} 3^{-2k} \cdot \frac{9^k}{k} \\
&\le \frac{4}{3} + \frac{1}{3} \sum_{k=1}^{t-1} \frac{1}{k} \\
&\le \log t.
\end{align*}
We have obtained the first inequality.

The second inequality follows easily from Cauchy--Schwarz. We turn now to the third inequality. In a similar fashion to \eqref{eq:cov-from-t-2-to-t-1}, we have
\begin{equation}\label{ineq:Cov-comb-1}
\Cov\!\left( X_{(i,j,t-i-j)}, X_{(i',j',t-i'-j')} \right) \ge \frac{1}{3} \sum_{k = \left\lceil \sqrt{t} \right\rceil}^{t-1} 9^{-k} \sum_{w \in \mathbb{Z}^3} \binom{k}{w} \binom{k}{w+\Delta},
\end{equation}
where $\Delta := (i-i',j-j',i'+j'-i-j)$.

Let us partition the index set $\{0,1,\ldots,k\}^3$ into three pieces: the first piece
\[
\mathcal{R} := \left[ \frac{k}{3} - k^{2/3}, \frac{k}{3} + k^{2/3} \right]^3 \cap \mathbb{Z}^3,
\]
the second piece
\[
\mathcal{R}' := [k]^3 - \mathcal{R},
\]
and the third piece $\mathcal{R}'' := \{0,1,\ldots,k\}^3 - (\mathcal{R} \cup \mathcal{R}')$.

Like in the proof of Lemma~\ref{l:sqrt-t-1-50-sqrt-t}, we have
\begin{equation}\label{ineq:Cov-comb-2}
\sum_{w \in \mathbb{Z}^3} \binom{k}{w} \binom{k}{w+\Delta} \ge \sum_{w \in \mathcal{R}} \binom{k}{w} \binom{k}{w+\Delta} \ge \frac{1}{5} \sum_{w \in \mathcal{R}} \binom{k}{w}^2.
\end{equation}
We proceed to bound the tail sum from above. For each $w \in \mathcal{R}'$, using an exact version of Stirling's approximation, we find
\[
\binom{k}{w} \le 2^{50} \cdot \sqrt{k} \cdot 3^k \cdot \exp\!\left( - \frac{9}{4} \sqrt[3]{k} \right),
\]
and thus
\[
\sum_{w \in \mathcal{R}'} \binom{k}{w}^2 \le k^3 \cdot 2^{100} \cdot k \cdot 9^k \exp\!\left( - \frac{9}{2} \sqrt[3]{k} \right) \le \frac{1}{2^{3000}} \cdot \frac{9^k}{k^2}.
\]
For each $w \in \mathcal{R}''$, at least one entry of $w$ is zero, and so $\binom{k}{w} \le 2^k$. This gives
\[
\sum_{w \in \mathcal{R}''} \binom{k}{w}^2 \le 3k \cdot 4^k \le \frac{1}{2^{10000}} \cdot \frac{9^k}{k^2}.
\]

Combining the above estimates with Lemma~\ref{lem:sqrt-27-4-pi} yields
\begin{align}
\sum_{w \in \mathcal{R}} \binom{k}{w}^2
&= \sum_{w \in \mathbb{Z}^3} \binom{k}{w}^2 - \sum_{w \in \mathcal{R}'} \binom{k}{w}^2 - \sum_{w \in \mathcal{R}''} \binom{k}{w}^2 \notag \\
&\ge \frac{1}{3} \cdot \frac{9^k}{k} - \frac{1}{2^{3000}} \cdot \frac{9^k}{k^2} - \frac{1}{2^{10000}} \cdot \frac{9^k}{k^2} \notag \\
&\ge \frac{1}{10} \cdot \frac{9^k}{k}. \label{ineq:Cov-comb-3}
\end{align}

Combining \eqref{ineq:Cov-comb-1}, \eqref{ineq:Cov-comb-2}, and \eqref{ineq:Cov-comb-3}, we obtain
\[
\Cov\!\left( X_{(i,j,t-i-j)}, X_{(i',j',t-i'-j')} \right) \ge \frac{1}{150} \sum_{k=\left\lceil \sqrt{t} \right\rceil}^{t-1} \frac{1}{k} \ge \frac{\log t}{400},
\]
as desired.
\end{proof}

\section{Some technical results on enumerating abelian squares}\label{s:abelian}
This section discusses some results about the numbers
\[
F(m,i) := \sum_{\substack{v_1, \ldots, v_i \in \mathbb{Z}_{\ge 0} \\ v_1 + \cdots + v_i = m}} \binom{m}{v_1, \ldots, v_i}^2,
\]
for positive integers $m$ and $i$. The integer $F(m,i)$ enumerates the {\em abelian squares} of length $2m$ in the alphabet of $i$ letters. We refer to the work of Richmond--Shallit \cite{RS09} for details. Richmond and Shallit \cite{RS09} show that for fixed $i$, we have
\begin{equation}\label{eq:F-m-i-asymp}
F(m,i) \sim \frac{i^{2m+i/2}}{(4\pi m)^{(i-1)/2}},
\end{equation}
as $m \to \infty$.

The numbers $F(m,i)$ are well-studied in combinatorics and number theory. By fixing a small value of $m$ or $i$ and considering the resulting sequence, we obtain many different sequences on the \cite{OEIS}. For example, \href{https://oeis.org/A000984}{A000984} ($i=2$), \href{https://oeis.org/A002893}{A002893} ($i=3$), \href{https://oeis.org/A002895}{A002895} ($i=4$), \href{https://oeis.org/A000384}{A000384} ($m=2$), \href{https://oeis.org/A169711}{A169711} ($m=3$), \href{https://oeis.org/A169712}{A169712} ($m=4$).

\medskip

For $i = 3$, we prove the following bound.

\begin{lemma}\label{lem:sqrt-27-4-pi}
For every positive integer $m$, we have
\[
\frac{1}{3} \cdot \frac{9^m}{m} \le F(m,3) \le \frac{\sqrt{27}}{4\pi} \cdot \frac{9^m}{m}.
\]
\end{lemma}
\begin{proof}
Let us define
\[
a_m := \frac{m}{9^m} \cdot F(m,3).
\]
The asymptotic formula \eqref{eq:F-m-i-asymp} shows that
\[
\lim_{m \to \infty} a_m = \frac{\sqrt{27}}{4\pi}.
\]
Since $a_1 = 1/3$, it suffices to show that $\{a_m\}_{m=1}^{\infty}$ is increasing.

The sequence $\{F(m,3)\}_{m=1}^{\infty}$ (which appears as \cite[A002893]{OEIS}) satisfies the following recurrence
\begin{equation}\label{eq:F-m-2-3}
(m+2)^2 F(m+2,3) - (10m^2+30m+23) F(m+1,3) + 9(m+1)^2 F(m,3) = 0,
\end{equation}
for every positive integer $m$. (See e.g. Matthijs Coster's comment on \cite[A002893]{OEIS}. One way to obtain and prove this recurrence is by using Zeilberger's {\em creative telescoping} algorithm---see \cite{PWZ96}. We refer to the OEIS page for existing formulas and pointers to the literature about this sequence.)

The recurrence \eqref{eq:F-m-2-3} implies that
\begin{equation}\label{eq:9m-m-1-m-2}
9m(m+1)(m+2)(a_{m+2} - a_{m+1}) = (m^3+3m^2+5m)(a_{m+1} - a_m) + (2m-1) a_m,
\end{equation}
for every positive integer $m$. Since $a_2 - a_1 = 1/27 > 0$ and $a_m > 0$ for every $m \ge 1$, we obtain from \eqref{eq:9m-m-1-m-2} by induction that $a_{m+1} - a_m > 0$ for every positive integer $m$.
\end{proof}

Note that our argument in the proof also shows that the constants $\frac{1}{3}$ and $\frac{\sqrt{27}}{4\pi}$ in the statement of Lemma~\ref{lem:sqrt-27-4-pi} are optimal.

For general $i$, we have the following bound, which is quite loose but is adequate for our proof of Lemma~\ref{lem:i-2-l-j-i}.

\begin{lemma}\label{l:general-i}
For positive integers $n$ and $i$ such that $n \ge (2i)^{100}$, we have
\[
F(n,i) \le \frac{i^{2n+i}}{n^{(i-1)/2}}.
\]
\end{lemma}
\begin{proof}
Let us first discuss the cases when $i$ is small. When $i = 1$, we have $F(n,1) = 1$. When $i = 2$, we have $F(n,i) = \binom{2n}{n}$, and the result follows easily by an exact version of Stirling's formula. The case when $i = 3$ follows from Lemma~\ref{lem:sqrt-27-4-pi} above.

For the rest of this proof, let us assume $i \ge 4$. The desired inequality can be proved in a similar manner as the argument in \cite{RS09}. Namely, we first break the sum
\[
\sum_{\substack{n_1, \ldots, n_i \in \mathbb{Z}_{\ge 0} \\ n_1 + \cdots + n_i = n}} \binom{n}{n_1, \ldots, n_i}^2
\]
into two parts: the part where $\forall j \in [i], |n/i - n_j| \le n^{5/8}$, and the rest. The latter is exponentially small with respect to the former, and the former can be approximated by an integral after the scaling $n_j = n/i + x_j \sqrt{n}$ with $|x_j| \le n^{1/8}$. We include details as follows.

Let us consider the following set
\[
\mathcal{R} := \left\{ (n_1, \ldots, n_i) \in \mathbb{Z}_{\ge 0}^i \, : \, n_1 + \cdots + n_i = n \text{ and } \forall j \in [i], \left| \frac{n}{i} - n_j \right| \le n^{5/8} \right\},
\]
and let $\mathcal{R}'$ denote the other $i$-tuples (where there is an index $j \in [i]$ such that $|n/i - n_j| > n^{5/8}$). Let us denote by $\Gamma:(0,\infty) \to \mathbb{R}$ the Gamma function defined on the positive reals. Using (i)~the convexity of $\log \Gamma$ (the digamma function $\psi$ is strictly increasing on $(0,\infty)$), (ii)~an exact version of Stirling's approximation for the Gamma function (see e.g. the recent sharp exact result of Nemes~\cite{Nem15}), and (iii)~our assumption that $n \ge (2i)^{100}$, we can show that for every $(n_1, \ldots, n_i) \in \mathcal{R}'$, we have
\[
\binom{n}{n_1, \ldots, n_i}^2 \le \frac{i^{2n+i}}{(2\pi n)^{i-1}} \cdot \frac{2}{e^{i \sqrt[4]{n}}}.
\]
Since $|\mathcal{R}'| \le (2n)^i$, we conclude
\begin{equation}\label{ineq:n1-ni-R'}
\sum_{(n_1, \ldots, n_i) \in \mathcal{R}'} \binom{n}{n_1, \ldots, n_i}^2 \le \frac{1}{2^{1000}} \cdot \frac{i^{2n+i}}{n^{100(i-1)}}.
\end{equation}

Next, we provide an upper bound for the sum
\[
\sum_{(n_1, \ldots, n_i) \in \mathcal{R}} \binom{n}{n_1, \ldots, n_i}^2.
\]
We take the same approach as~\cite{RS09} of approximating the sum above by an integral. However, since we would like to obtain an exact bound, we have to be particularly careful about the error term. Let
\[
I := \left[ \frac{n}{i} - n^{5/8}, \frac{n}{i} + n^{5/8} \right] \cap \mathbb{Z}.
\]
We start by the trivial upper bound
\begin{equation}\label{ineq:I-i-1}
\sum_{(n_1, \ldots, n_i) \in \mathcal{R}} \binom{n}{n_1, \ldots, n_i}^2 \le \sum_{(n_1, \ldots, n_{i-1}) \in I^{i-1}} \frac{\Gamma(n+1)^2}{\Gamma(n_1+1)^2 \cdots \Gamma(n_i + 1)^2},
\end{equation}
where $n_i$ on the right-hand side denotes $n - n_1 - \cdots - n_{i-1}$.

For $n_1, \ldots, n_{i-1} \in I$ and $\varepsilon_1, \ldots, \varepsilon_{i-1} \in [-1,1]$, let us define
\[
K := K(n_1, \ldots, n_{i-1}; \varepsilon_1, \ldots, \varepsilon_{i-1})
\]
to be the expression
\[
K := \frac{\Gamma(n_1+1)^2 \cdots \Gamma(n_i+1)^2}{\Gamma(n_1 + \varepsilon_1+1)^2 \cdots \Gamma(n_i + \varepsilon_i + 1)^2},
\]
where $n_i := n - n_1 - \cdots - n_{i-1}$ and $\varepsilon_i = - \varepsilon_1 - \cdots - \varepsilon_{i-1}$. Using an exact bound on the Gamma function (see e.g. Nemes' inequality~\cite{Nem15} again), we find that
\[
K \ge \left( \frac{n_i}{n_1} \right)^{2\varepsilon_1} \left( \frac{n_i}{n_2} \right)^{2\varepsilon_2} \cdots \left( \frac{n_i}{n_{i-1}} \right)^{2\varepsilon_{i-1}} \cdot \exp\!\left( - \frac{10 i^3}{n} \right).
\]
Since for every $\lambda > 0$, we have\footnote{Of course, the integral on the left-hand side of~\eqref{ineq:int-AM-GM} can be computed explicitly rather easily, but we choose to present this nice quick argument.}
\begin{equation}\label{ineq:int-AM-GM}
\int_{-1}^1 \lambda^\varepsilon \, {\rm d} \varepsilon = \int_{-1}^1 \frac{\lambda^\varepsilon + \lambda^{-\varepsilon}}{2} \, {\rm d} \varepsilon \overset{\text{(AM-GM)}}{\ge} 2,
\end{equation}
we find that
\[
\int_{-1}^1 \cdots \int_{-1}^1 K \, {\rm d}\varepsilon_1 \cdots {\rm d}\varepsilon_{i-1} \ge 2^{i-1} \exp\!\left( - \frac{10i^3}{n} \right),
\]
and therefore
\begin{equation}\label{ineq:fint}
\fint_{[-1,1]^{i-1}} K \, {\rm d}\varepsilon_1 \cdots {\rm d}\varepsilon_{i-1} \ge  \exp\!\left( - \frac{10i^3}{n} \right).
\end{equation}

Combining~\eqref{ineq:I-i-1} and~\eqref{ineq:fint}, we find that
\begin{align*}
&\sum_{(n_1, \ldots, n_i) \in \mathcal{R}} \binom{n}{n_1, \ldots, n_i}^2 \\
&\le \sum_{(n_1, \ldots, n_i) \in I^{i-1}} \Gamma(n+1)^2 \cdot \frac{\exp(10i^3/n)}{2^{i-1}} \cdot \int_{[-1,1]^{i-1}} \frac{{\rm d}\varepsilon_1 \cdots {\rm d}\varepsilon_{i-1}}{\Gamma(n_1+\varepsilon_1+1)^2 \cdots \Gamma(n_i + \varepsilon_i + 1)^2} \\
&\le \exp\!\left( \frac{10i^3}{n} \right) \int_{{\widetilde{I}}^{i-1}} \frac{\Gamma(n+1)^2}{\Gamma(y_1+1)^2 \cdots \Gamma(y_i+1)^2} {\rm d} y_1 \cdots {\rm d} y_{i-1},
\end{align*}
where $y_i := n - y_1 - \cdots - y_{i-1}$, and
\[
\widetilde{I} := \left[ \frac{n}{i} - n^{5/8} - 1, \frac{n}{i} + n^{5/8} + 1 \right].
\]
Using an exact approximation of the Gamma function again, we find that for $y_1, \ldots, y_{i-1} \in \widetilde{I}$, the integrand can be estimated as
\[
\frac{\Gamma(n+1)^2}{\Gamma(y_1+1)^2 \cdots \Gamma(y_i+1)^2} \le \frac{i^{2n+i}}{(2\pi n)^{i-1}} \exp\!\left( \frac{5i^3}{n^{1/8}} \right) \cdot \exp\!\left(-i(x_1^2 + \cdots + x_i^2)\right),
\]
where $x_1, \ldots, x_i$ are given by $y_j = \frac{n}{i} + x_j \sqrt{n}$. Therefore,
\[
\sum_{(n_1, \ldots, n_i) \in \mathcal{R}} \binom{n}{n_1, \ldots, n_i}^2 \le \exp\!\left( \frac{10i^3}{n} \right) \cdot \frac{i^{2n+i}}{(2\pi n)^{i-1}} \cdot \exp\!\left( \frac{5i^3}{n^{1/8}} \right) \cdot \sqrt{n}^{i-1} \cdot A,
\]
where
\[
A := \int_{\mathbb{R}^{n-1}} \exp\!\left( - i\left( x_1^2 + \cdots + x_{i-1}^2 + \left(-x_1-x_2-\cdots-x_{i-1}\right)^2 \right) \right) \, {\rm d} x_1 \cdots {\rm d} x_{i-1}.
\]
This integral $A$ has been computed explicitly in \cite{RS09}:
\[
A = \pi^{(i-1)/2} i^{-i/2}.
\]
Hence,
\begin{equation}\label{ineq:n1-ni-R-main}
\sum_{(n_1, \ldots, n_i) \in \mathcal{R}} \binom{n}{n_1, \ldots, n_i}^2 \le \exp\!\left( \frac{6i^3}{n^{1/8}} \right) \cdot \frac{i^{2n+i/2}}{(4\pi n)^{(i-1)/2}}.
\end{equation}
We finish by combining~\eqref{ineq:n1-ni-R'} and \eqref{ineq:n1-ni-R-main}. Note that we have actually derived a much tighter upper bound than the one in the lemma statement.
\end{proof}

\section{Open Questions}
Here we list some ideas for further investigation.

\subsection{Local reconstruction for the finite model}\label{subsec:local-finite}
In the finite model $P = \mathbb{Z}_{\ge 0}^{d+1}$, our present work does not answer when local, local convex, or single-vertex reconstruction is possible for general dimension $d$. Indeed, our work gives a partial answer that none of these types of reconstruction is possible if $(\alpha_1, \ldots, \alpha_{d+1})$ is inside the box
\[
[0,1] \times [0,1/2] \times \cdots \times [0,1/(d+1)].
\]
We also know (see the discussion at the beginning of Section~\ref{s:orthant}) that if $\alpha_1 > 1$, then single-vertex reconstruction is possible by simply taking $\zeta = X_{(t,0,\ldots,0)}$. What can we say about other cases of $(\alpha_1, \ldots, \alpha_{d+1})$?

\subsection{Better bounds for the rate of convergence}
In Subsection~\ref{subsec:in-box}, we showed that reconstruction is not possible by establishing
\[
S_t = \sup_\zeta \frac{\Cov(\zeta, X_0)}{\sqrt{\Var(\zeta)}} \le \frac{C}{\log^{(2d+2)}(t)}.
\]
(See Theorem~\ref{thm:corr-log-log-log-log-t}.) While this shows that $S_t \to 0$ as $t \to \infty$, we believe that the right-hand side above is not the correct rate of convergence.

In the critical case when $d = 1$, we saw in Theorem~\ref{thm:exact-formula} that $S_t \sim \sqrt[4]{\pi/t}$. What are the analogous asymptotic formulas in other cases?

\subsection{Explicit formulas}
Theorem~\ref{thm:exact-formula} gives the exact formula
\[
S_t = \frac{1}{\sqrt{\frac{t}{4^t}\binom{2t}{t} + 1}}
\]
in the critical case for $d = 1$ in the finite model. Furthermore, it gives the explicit formula for an optimal estimator $\zeta$ for the $t^{\text{th}}$ layer. (See Remark~\ref{rem:exact-zeta}.) It may be of interest to see similar exact computations for other cases.

\subsection{Reconstruction for other posets}
Our definition of reconstruction in Section~\ref{s:defns-main-results} works for other posets besides the two main models we consider in this paper. It may be interesting to ask whether each type of reconstruction is possible in other posets.

\subsection{Other model recurrences}
In the present paper, our model recurrence \eqref{eq:model-rec} combines the random variables from the parent nodes and the Gaussian noises additively. One might be interested in considering a different model recurrence where a different operation is used to compute the random variable at each node and then asking similar reconstruction questions.

\bigskip

\section*{Acknowledgments}
P.J. was supported by Elchanan Mossel's Vannevar Bush Faculty Fellowship ONR-N00014-20-1-2826 and by Elchanan Mossel's Simons Investigator award (622132). 
E.M. was partially supported by Bush Faculty Fellowship ONR-N00014-20-1-2826, by Simons Investigator award (622132) and by grant NSF DMS-2031883.

We would like to thank Chenghao Guo, Han Huang, Sung Woo Jeong, Sergei Korotkikh, Yury Polyanskiy, Thana Somsirivattana, Wijit Yangjit, and Yuan Yao for discussions. We thank Sorawee Porncharoenwase for useful discussions and computations. We used \texttt{Desmos}, \texttt{Maple}, \texttt{R}, \texttt{Racket}, \texttt{WolframAlpha}, and the Maple program \texttt{EKHAD} written by Doron Zeilberger to help with computations.

\bibliographystyle{alpha}
\bibliography{ref,my,all}

\end{document}